\documentclass[aps,prl,reprint,superscriptaddress,nofootinbib,floatfix]{revtex4-2}
\pdfoutput=1

\usepackage{amssymb}
\usepackage{amsmath}
\usepackage{graphicx}
\usepackage{tikz}
\usepackage[T1]{fontenc}
\usepackage[utf8]{inputenc}

\xdefinecolor{mylinkcolor}{rgb}{0,0,0.7}
\usepackage[
  colorlinks=true, citecolor=mylinkcolor,
  linkcolor=mylinkcolor, urlcolor=mylinkcolor,
]{hyperref}

\allowdisplaybreaks

\usepackage{amsthm}
\newtheorem{theorem}{Theorem}
\newtheorem{corollary}{Corollary}

\newcommand{\g}{\mathfrak{g}}
\newcommand{\gC}{\mathfrak{g}_\mathbb{C}}
\newcommand{\gR}{\mathfrak{g}_\mathbb{R}}
\newcommand{\rad}{\mathfrak{r}}
\newcommand{\gss}{\mathfrak{g}_{\mathrm{ss}}}
\newcommand{\ad}{\mathrm{ad}}
\newcommand{\Sym}{\mathrm{Sym}}
\newcommand{\Tr}{\mathrm{Tr}}

\begin{document}

\title{Lie Algebra Saddles in the IKKT Matrix Model\\ and\\ Criteria for the Emergence of Time}

\author{Henry Liao}
\affiliation{Department of Physics and Center for Theoretical Physics, National Taiwan University, Taipei 10617, Taiwan}
\email{henryliao.physics@gmail.com}

\begin{abstract}
In the IKKT matrix model, the dimension, geometry, and metric signature of spacetime must emerge from the matrix dynamics.
We classify the Lie-algebraic saddles of the model: classical solutions in which the ten Hermitian matrices realize a finite-dimensional real Lie algebra in a unitary representation.
For this class the equations of motion collapse to a single algebraic equation for a symmetric bilinear form $h$ on the algebra, and the signature of the emergent spacetime can be read off from $h$.
Two facts control the outcome.
First, when the complexified algebra is simple, a Casimir obstruction allows only $h=0$; this explains why fuzzy spheres and the expanding-universe solutions need mass terms or infrared regulators.
Second, semisimple algebras admit saddles only through couplings between distinct simple ideals; such saddles first appear in six dimensions and reach every signature only from twelve.
These facts yield criteria for the emergence of time.
Every nondegenerate Lorentzian saddle in at most eleven dimensions---hence every Lorentzian background of the ten-dimensional model---has a nonzero solvable radical. 
The only nontrivial possibility of the semisimple part is the Lorentz algebra $\mathfrak{so}(1,3)$, which is purely spacelike, and every timelike direction has a component in the radical.
When the semisimple part is nontrivial, it acts on the radical only from ten dimensions on; at ten dimensions, the radical must be abelian and must carry the Weyl-spinor representation, which solves the Lorentzian equations of motion; a nonabelian radical first appears in eleven dimensions.
As a result, space can be semisimple; time must be solvable.
\end{abstract}

\maketitle

\section{Introduction}\label{sec:intro}
The IKKT matrix model~\cite{Ishibashi:1996xs} is a candidate nonperturbative definition of the type IIB superstring in which spacetime is not part of the definition of the theory but emerges dynamically.
Its degrees of freedom are ten Hermitian bosonic matrices $A^\mu$ and their fermionic superpartners.
Geometry---the dimensionality, shape, and extent of the emergent space, viewed as a brane whose embedding the $A^\mu$ encode---must arise from their dynamics~\cite{Aoki:1998vn}.
The original proposal used Euclidean signature~\cite{Ishibashi:1996xs}; later works study Lorentzian versions of the model~\cite{Kim:2011cr,Kim:2011ts,Kim:2012mw,Ito2014rg,Ito2015pos,Steinacker:2017vqw,SternXu2018,NishimuraTsuchiya2019cl,Aoki2019str,SperlingSteinacker2019,Hatakeyama2020cls,Klinkhamer2020emu,Klinkhamer2020mf,Klinkhamer2020pts,Klinkhamer2020rbb,Klinkhamer2021mfe,Nishimura:2022alt,Anagnostopoulos:2022dak,Hirasawa2023reg,BattistaSteinacker2023,SteinackerTran2023lor,Tsuchiya2024reg,Anagnostopoulos2026cl,Anagnostopoulos2026susy}; Ref.~\cite{Steinacker:2026inv} reviews the Lorentzian program.
The Euclidean and Lorentzian versions of the model behave differently.
Lorentzian Monte Carlo studies find evidence for an expanding $(3{+}1)$-dimensional universe~\cite{Kim:2011cr,Anagnostopoulos:2022dak}, and the classical expanding solutions behind this behavior exist only in the infrared-regularized model~\cite{Kim:2011ts,Kim:2012mw}.
In the Euclidean model, no nontrivial noncommutative background survives at finite matrix size~\cite{Steinacker:2017vqw}.
Simulations, classical solutions, and deformations have probed both versions intensively~\cite{Aoki:1998vn,Kim:2011cr,Kim:2011ts,Kim:2012mw,Steinacker:2017vqw,Nishimura:2022alt,Anagnostopoulos:2022dak,Tsuchiya2024reg}.
What is still missing is a structural statement of which signatures the classical solutions can realize and where, inside a configuration, the timelike direction can reside.

This work analyzes the class of \emph{Lie-algebraic} backgrounds: configurations whose matrices realize a finite-dimensional real Lie algebra $\gR$ in a unitary---generically infinite-dimensional---representation.
A Lie algebra is a vector space equipped with an antisymmetric bracket $[\cdot,\cdot]$ that satisfies the Jacobi identity; for matrices, the bracket is the commutator.
The ansatz therefore asks that the commutator of any two of the matrices $A^\mu$ be a linear combination of the matrices themselves.
Every finite-dimensional Lie algebra decomposes, by the Levi decomposition $\g\cong\gss\ltimes\rad$, into a semisimple part $\gss$ acting on the radical $\rad$, the maximal solvable ideal.
A semisimple algebra is a direct sum of simple ideals; the prototype is angular momentum, $[J_i,J_j]=i\epsilon_{ijk}J_k$, which generates $\mathfrak{su}(2)$, the algebra of the fuzzy sphere~\cite{Madore:1991bw}.
A solvable algebra is one whose derived series $\g\supset[\g,\g]\supset[[\g,\g],[\g,\g]]\supset\cdots$ terminates at zero; the prototype is the canonical pair, $[\hat x,\hat p]=i$, which generates the Heisenberg algebra, the algebra of the Moyal plane~\cite{Douglas:2001ba}.
Loosely speaking, semisimple algebras generate fuzzy spheres and hyperboloids, and solvable algebras generate flat, Moyal-type geometries.
We address two questions.
\begin{itemize}
\item[(a)] For the model with Euclidean or with Lorentzian target signature, which solutions does the Lie-algebraic ansatz admit?
\item[(b)] For the Lorentzian version, where does the timelike direction reside among the generators of a solution?
\end{itemize}

The Lie-algebraic class contains the standard algebraic backgrounds: Moyal planes (once the identity produced by the commutators is counted as a central generator of the Heisenberg algebra), fuzzy geometries, the expanding-universe solutions of Refs.~\cite{Kim:2011ts,Kim:2012mw}, and the compactified brane solutions with split noncommutativity of Ref.~\cite{Steinacker:2011wb}.
It excludes Snyder-type embeddings~\cite{Snyder:1946qz,SperlingSteinacker2019}, in which the matrices sit inside a larger algebra without closing among themselves; those satisfy weaker equations, and Sec.~\ref{sec:disc} discusses them.
It is likewise complementary to the covariant formulations, in which the matrices realize infinite-dimensional operator algebras: covariant derivatives on curved manifolds~\cite{HanadaKawaiKimura2006,Tsuchiya2024reg}, or bilocal fields reproducing general relativity at low energies~\cite{HoKawaiSteinacker2026}.
The results below also account for the algebraic form of those constructions.
Within the class itself, Chatzistavrakidis~\cite{Chatzistavrakidis:2011su} scanned the classification tables of all real Lie algebras of dimension $\le5$ and of all nilpotent ones of dimension $6$, placing the generators directly along coordinate axes, so that the metric the algebra inherits is the flat one in the tabulated basis.
Every solution found there is nilpotent or solvable.
Our results explain this outcome and extend the search: we keep the algebra fixed but let its embedding into the target space vary, so that the inherited metric becomes an unknown to be solved for, and we prove that no algebra of dimension $\le5$ with a nonzero semisimple part can host a background of the Euclidean or of the Lorentzian model.
The first semisimple solution appears at exactly six dimensions---the $\mathfrak{so}(1,3)$ solution of Ref.~\cite{Liao:2025yfb}, beyond that survey.

On this class the equations of motion reduce to a single algebraic equation, Eq.~\eqref{eq:master}, for a real symmetric bilinear form $h$ on the algebra.
The form $h$ is the target-space metric $\eta$ pulled back to the algebra, and $\eta$ enters only through $h$.
We therefore do not fix the signature of $\eta$ in advance; we solve Eq.~\eqref{eq:master} for $h$ and read the signature off the solutions.
Throughout, we take positive directions of $h$ to be timelike; Lorentzian means signature $(1,d-1)$, with $d$ the dimension of the algebra.

Our results are the following.

(1) On a real Lie algebra whose complexification is simple, the only solution is $h=0$ (Theorem~\ref{thm:nogo}).
The obstruction comes from the quadratic Casimir.
It explains why fuzzy-sphere backgrounds exist only in mass- or Myers-deformed models~\cite{Myers:1999ps,Iso:2001mg,Steinacker:2010rh,GoharaSako2025}, and why the expanding solutions of Refs.~\cite{Kim:2011ts,Kim:2012mw} require infrared regulators, whose multipliers are fixed by the adjoint Casimir.

(2) On a semisimple algebra every solution vanishes on each simple ideal and is carried entirely by the couplings between distinct ideals, that is, by the off-block-diagonal components of $h$ (Corollary~\ref{cor:semisimple}).
These statements concern the complexified equation.
Real Lie algebras are its real forms, selected by a conjugation, and the reality condition is the compatibility of $h$ with that conjugation (Sec.~\ref{sec:real}).
Once we impose it, the attainable signatures fall into three cases (Theorem~\ref{thm:signature}), according to how the conjugation acts on the simple ideals: an ideal mapped to itself forces at least three positive and three negative directions, a single pair of ideals swapped by the conjugation contributes an even number of positive directions, and two or more such pairs---possible only for $d\ge12$---attain every signature.

(3) Consequently, every nondegenerate Lorentzian solution with $d\le11$ has a nonzero radical (Theorem~\ref{thm:main}).
Euclidean solutions, by contrast, exist already on the semisimple sector, starting from the six-dimensional solution above.
Moreover, the semisimple part is zero or a single copy of $\mathfrak{so}(1,3)$, and $h$ is negative definite on it: every nonzero vector of the semisimple part is spacelike, and every timelike or null vector has a nonzero component in the radical (Corollary~\ref{cor:radical}).

(4) Representation theory decides whether the semisimple part can act nontrivially on the radical, $[\gss,\rad]\neq0$ (Theorem~\ref{thm:semi}).
For $d\le9$ it cannot, and every solution lives on a direct sum $\gss\oplus\rad$; the first solutions on a genuine semidirect sum appear at $d=10$, where the radical must be abelian and must transform as a Weyl spinor, and a nonabelian radical first appears at $d=11$.
Together with the low-dimensional classification of Ref.~\cite{Chatzistavrakidis:2011su}, these results organize the known solutions, which Appendix~\ref{app:F} collects.

In the ten-dimensional model of either signature, the semisimple part of any Lie-algebraic background---nondegenerate or not---is $0$ or $\mathfrak{so}(1,3)$, with $h$ negative definite on it (Sec.~\ref{sec:lorentz}, Appendix~\ref{app:C}); the Euclidean saddle of Ref.~\cite{Liao:2025yfb} is therefore the only semisimple building block available in ten dimensions.
All statements are classical: we solve the classical equation of motion exactly, without regularization, and we do not address the quantum dynamics.
How to extract an emergent geometry from a given solution remains a separate, open problem, which we discuss briefly in Sec.~\ref{sec:disc}.
In summary, in every Lie-algebraic background of the ten-dimensional model, the semisimple sector cannot carry a timelike direction; time lies in the solvable radical.

The paper is organized as follows.
Section~\ref{sec:setup} defines the ansatz and reduces the equation of motion to a finite condition, Eq.~\eqref{eq:master}.
Section~\ref{sec:class} classifies the solutions: over the complex numbers, then for the real forms and their signatures, then for Lorentzian signature.
Section~\ref{sec:radical} constructs explicit solutions, on direct and on semidirect sums.
Section~\ref{sec:disc} relates the results to other approaches and lists the open problems.
The appendices contain the proofs (A, B, and G), the pullback analysis (C), the explicit solutions (D), the infrared regulators (E), and the solution landscape (F).

\section{Setup and strategy}\label{sec:setup}
The degrees of freedom of the IKKT model are ten Hermitian bosonic matrices $A^\mu$ and a Majorana--Weyl spinor $\Psi$ of $\mathrm{SO}(1,9)$~\cite{Ishibashi:1996xs}.
The fermions enter the bosonic equations of motion only bilinearly, so we can consistently set $\Psi=0$, and we do so throughout.
The action of the bosonic sector is
\begin{equation}\label{eq:SYM}
  S=-\tfrac14\,\Tr\!\big(\eta_{\mu\rho}\eta_{\nu\sigma}[A^\mu,A^\nu][A^\rho,A^\sigma]\big),
\end{equation}
with $\mu,\nu,\rho,\sigma=1,\dots,k$.
We do not fix the dimension $k$ or the signature of the real symmetric metric $\eta$ in advance; for the type IIB superstring, $k=10$, and $\eta$ is Euclidean or Lorentzian.
The classification below therefore applies to every Yang--Mills-type matrix model, for example the BFSS matrix model~\cite{Banks:1996vh}, and Sec.~\ref{sec:disc} spells this out.
The equations of motion for $A^\mu$ are
\begin{equation}\label{eq:EOM}
  \eta_{\nu\rho}\,[A^\nu,[A^\rho,A^\mu]]=0 .
\end{equation}

We now state the ansatz.
A configuration is \emph{Lie-algebraic} if it is built on a finite-dimensional real Lie algebra $\gR$.
With $\{t^a\}$, $a=1,\dots,d$, a basis of $\gR$ and $C$ a real $k\times d$ matrix of rank $d\le k$, the ansatz is
\begin{equation}\label{eq:ansatz}
  A^\mu=C^\mu{}_a\,t^a ,
\end{equation}
with the commutators in Eq.~\eqref{eq:EOM} read as Lie brackets.
In words, the linear span of the matrices $A^\mu$ closes under the commutator, and $C$ records which combination of generators occupies which target-space direction.
We recover matrices as $A^\mu=i\,C^\mu{}_a\,\pi(t^a)$, with $\pi$ a faithful unitary representation of $\gR$; such a representation exists for every real Lie algebra~\cite{Knapp2002}.
The generators of a unitary representation satisfy $\pi(x)^\dagger=-\pi(x)$, so the matrices $A^\mu$ are Hermitian.
For noncompact real forms and for nonabelian nilpotent or solvable algebras, such representations are necessarily infinite-dimensional.
Since $\pi$ preserves brackets and is faithful, and the powers of $i$ give only an overall constant in Eq.~\eqref{eq:EOM}, working with the Lie bracket loses no generality.

Substituting Eq.~\eqref{eq:ansatz} into Eq.~\eqref{eq:EOM} gives
\begin{equation}\label{eq:EOMansatz}
  {C^\mu{}_c\;h_{ab}\,[t^a,[t^b,t^c]]=0\quad\forall\,\mu,\qquad h=C^{T}\eta\,C .}
\end{equation}
Since $C$ has rank $d$, it is injective on the index $c$, so Eq.~\eqref{eq:EOM} is equivalent to
\begin{equation}\label{eq:master}
  h_{ab}\,[t^a,[t^b,t^c]]=0\quad\forall\,c,
  \qquad h=C^{T}\eta\,C .
\end{equation}

In the Euclidean model at finite $N$, contracting Eq.~\eqref{eq:EOM} with $A_\mu$ and taking the trace forces all commutators of Hermitian matrices to vanish, so no noncommutative background survives~\cite{Steinacker:2017vqw}.
This argument uses the trace and does not apply to infinite-dimensional representations.
Equation~\eqref{eq:EOM} itself is an identity between operators and contains no trace, so it is well defined without regularization.
The action~\eqref{eq:SYM}, on the other hand, does require a regularized trace; we do not deal with that here.
The solutions below are exact solutions of Eq.~\eqref{eq:EOM}, with the same status as the Moyal plane and the expanding-universe backgrounds.

The metric $\eta$ enters Eq.~\eqref{eq:master} only through $h$, which is $\eta$ pulled back to the algebra: for $v\in\gR$, $h(v,v)=\eta(Cv,Cv)$ is the squared length of the target-space direction that $v$ occupies.
We therefore do not fix the signature of $\eta$ in advance; we solve Eq.~\eqref{eq:master} for $h$ and read the signature off.
Throughout, $p$ and $q$ count the positive and negative eigenvalues of $h$, positive directions are timelike, and Lorentzian means $(p,q)=(1,d-1)$; a Euclidean target has all directions spacelike, $(p,q)=(0,d)$.
Since Eq.~\eqref{eq:master} is linear in $h$, the overall sign of $\eta$ is immaterial: $h$ and $-h$ solve together.
For $d=k$ the correspondence is exact: $h$ is nondegenerate with the signature of $\eta$, and, by Sylvester's law of inertia, every nondegenerate solution of signature $(p,q)$ arises from some invertible $C$ with $\eta$ of that signature.
For $d<k$, $h$ is the pullback of $\eta$ along the injective $C$; it can be degenerate, but $p(h)\le p(\eta)$ and $q(h)\le q(\eta)$, and its null directions are bounded by the Witt index of $\eta$ (Appendix~\ref{app:C}).

\textit{Strategy.}---For each fixed algebra, Eq.~\eqref{eq:master} is a finite system of linear equations for $h$.
Guided by the Levi decomposition, we solve the semisimple part first and treat the radical afterwards.
Section~\ref{sec:complex} solves Eq.~\eqref{eq:master} over $\mathbb{C}$, where the classification of complex simple Lie algebras is available~\cite{Knapp2002} and where we prove the no-go theorem, Theorem~\ref{thm:nogo} (Appendix~\ref{app:A}).
The outcome is that only the off-block-diagonal parts of $h$, connecting distinct simple ideals, survive.
Section~\ref{sec:real} returns to the real algebras.
A real Lie algebra is its complexification with a reality condition; since the coefficients of Eq.~\eqref{eq:master} are real, the real solutions are exactly the complex solutions compatible with that condition (Appendix~\ref{app:B}).
This step determines the attainable signatures of the semisimple part.
Section~\ref{sec:lorentz} combines the semisimple part with the radical: a Lorentzian solution with $d\le11$ requires a nonzero radical.
Section~\ref{sec:radical} constructs explicit Lorentzian solutions, with the timelike direction in the radical, and determines when the semisimple part can act on the radical.

\section{Classification of the solutions}\label{sec:class}

\subsection{The complex problem: simple and semisimple algebras}\label{sec:complex}
Following the strategy, this subsection works over $\mathbb{C}$.
\begin{theorem}[Simple no-go]\label{thm:nogo}
Let $\gC$ be a complex simple Lie algebra.
Then the only solution of Eq.~\eqref{eq:master} is $h=0$.
\end{theorem}




Here we explain the mechanism in the language of coupled angular momenta; Appendix~\ref{app:A} gives the complete proof.
Write $K_H\equiv h_{ab}\,\ad_{t^a}\ad_{t^b}$, so that Eq.~\eqref{eq:master} states that $K_H=0$ as an operator on $\g$.

The case $\g=\mathfrak{su}(2)$, with $[J_a,J_b]=\epsilon_{abc}J_c$, can be done by hand: two lines of index algebra give $h_{ab}[J_a,[J_b,J_c]]=(h-\mathrm{tr}(h)\,\mathbb{I})_{ce}J_e$, so $K_H=0$ requires $h=\mathrm{tr}(h)\,\mathbb{I}$, whose trace gives $\mathrm{tr}\,h=3\,\mathrm{tr}\,h$ and hence $h=0$.
The trace part of $h$ and its traceless part are each multiplied by a nonzero number, and they cannot cancel.

The general case has the same structure.
Regard the symmetric tensor $h_{ab}$ as the wave function of two particles that both carry the adjoint representation, with generators $\vec J_1$ and $\vec J_2$.
In a basis orthonormal for the Killing form, the double commutator is the scalar product of the two spins acting on this wave function, $K_H=-\vec J_1\!\cdot\!\vec J_2\,h$ (Eq.~\eqref{eq:app-JJ}), so Eq.~\eqref{eq:master} asks for a two-spin state with vanishing spin--spin correlation.
Since $2\vec J_1\!\cdot\!\vec J_2=\vec J_{\rm tot}^{\,2}-\vec J_1^{\,2}-\vec J_2^{\,2}$ and the total spin is quantized, on each irreducible component of $h$ with total Casimir $C_\lambda$, the correlation is the fixed number $\tfrac12(C_\lambda-2c_\g)$, where $c_\g=\vec J_1^{\,2}=\vec J_2^{\,2}$ is the adjoint Casimir, and it vanishes only if $C_\lambda=2c_\g$ exactly.
For $\mathfrak{su}(2)$ the symmetric part of $1\otimes1$ is $0\oplus2$; the correlation $\tfrac12[j(j+1)-4]$ equals $-2$ on the singlet---this is the choice $h\propto\kappa^{-1}$, the inverse Killing form, for which $K_H$ is the adjoint Casimir itself---and $+1$ on the quintet, and zero would require $j(j+1)=4$.
Table~\ref{tab:casimir} lists the ratios $C_\lambda/c_\g$ on $\Sym^2(\g)$ across the classification of complex simple Lie algebras.
The value $2$ never occurs, and this proves Theorem~\ref{thm:nogo}.

By the strategy of Sec.~\ref{sec:setup}, the same conclusion holds for every real Lie algebra whose complexification is simple, such as $\mathfrak{su}(2)$, $\mathfrak{sl}(2,\mathbb{R})$, and $\mathfrak{su}(1,1)$.
A real simple Lie algebra can, however, have a non-simple---in particular semisimple---complexification, such as $\mathfrak{so}(1,3)_{\mathbb{C}}\cong\mathfrak{so}(3)_{\mathbb{C}}\oplus\mathfrak{so}(3)_{\mathbb{C}}$; such an algebra falls under the semisimple case, Corollary~\ref{cor:semisimple} below, and its solutions reside in the coupling between the two ideals.
Indeed, $\mathfrak{so}(1,3)$ does solve Eq.~\eqref{eq:master}, with $h_{ab}\propto\delta_{ab}$ in the basis of rotations $J_i$ and boosts $K_i$~\cite{Liao:2025yfb}.

The no-go statement accounts for three known results.
The fuzzy sphere~\cite{Madore:1991bw}, built on $\mathfrak{su}(2)$, is not a solution of the undeformed model; it exists in the presence of a Myers term or a mass term~\cite{Myers:1999ps,Iso:2001mg,Steinacker:2010rh}.
In the mass-deformed model, solutions built on semisimple Lie algebras exist in general~\cite{GoharaSako2025}; the mass term supplies the nonzero eigenvalue that the undeformed equation forbids.
The expanding-universe solutions of Refs.~\cite{Kim:2011ts,Kim:2012mw} exist only in the infrared-regularized Lorentzian model.
On their $\mathfrak{su}(2)$ and $\mathfrak{su}(1,1)$ branches, $h$ is proportional to the inverse Killing form, so the left-hand side of Eq.~\eqref{eq:EOM} equals the adjoint Casimir, and the Lagrange multipliers that implement the cutoffs are fixed to this nonzero eigenvalue (Appendix~\ref{app:E}).
Their expanding branch, commuting space acted on by time through $[A^0,A^i]\propto A^i$, needs the spatial multiplier for a different reason: in the undeformed equation that algebra forces the time generator to be null, $h_{00}=h_{0i}=0$ (Sec.~\ref{sec:radical}).

Applying Theorem~\ref{thm:nogo} to each simple ideal of a semisimple algebra gives the following corollary.
\begin{corollary}[Semisimple blocks]\label{cor:semisimple}
Let $\gC$ be semisimple, that is, $\gC=\bigoplus_i\g_i$ with simple ideals $\g_i$, and write $h$ in blocks $h^{(ij)}$.
Then $h$ solves Eq.~\eqref{eq:master} if and only if every diagonal block vanishes, $h^{(ii)}=0$; the off-diagonal blocks are arbitrary.
\end{corollary}
\emph{Proof.}
Distinct ideals commute, so for $c\in\g_i$ only $a,b\in\g_i$ contribute to Eq.~\eqref{eq:master}, and Theorem~\ref{thm:nogo} forces $h^{(ii)}=0$; conversely, with all diagonal blocks zero every term vanishes. \hfill$\square$

So each simple ideal is a null block of $h$, while the couplings between different ideals are completely free.
This block structure determines which signatures survive once we impose the reality condition, as we show next.

\subsection{Real signatures for the semisimple subalgebra}\label{sec:real}
We now return to the real algebras.
Hermiticity of the matrices requires a unitary representation of a Lie algebra over $\mathbb{R}$ (Sec.~\ref{sec:setup}), so the physical solutions are the real ones, and we obtain them by imposing a reality condition on the complex solutions.

Such a condition can be formulated with a map $\sigma$ on $\gC$ that is antilinear, $\sigma(\lambda x)=\bar{\lambda}\,\sigma(x)$; involutive, $\sigma^{2}=\mathrm{id}$; and an automorphism, $\sigma([x,y])=[\sigma(x),\sigma(y)]$.
The fixed points of $\sigma$ form a real Lie algebra, and every real form $\gR$ of $\gC$ arises this way~\cite{Knapp2002}.
This mimics Hermiticity at the level of the Lie algebra: the map $X\mapsto X^{\dagger}$ on matrices is antilinear and involutive, and its fixed points are the Hermitian matrices.
For example, on $\mathfrak{sl}(2,\mathbb{C})$, the choice $\sigma(x)=-x^{\dagger}$ fixes $\mathfrak{su}(2)$, and $\sigma(x)=\bar{x}$ fixes $\mathfrak{sl}(2,\mathbb{R})$.

A complex solution $h$ of Eq.~\eqref{eq:master} defines a real symmetric form on $\gR$ if and only if it is $\sigma$-compatible; conversely, every real solution arises this way (Appendix~\ref{app:B}).
The condition $\sigma^{2}=\mathrm{id}$ constrains how $\sigma$ acts on the ideals.
A semisimple $\gC$ is the direct sum of its simple ideals, $\gC=\bigoplus_i\g_i$; an automorphism maps simple ideals to simple ideals, so $\sigma$ acts on the set $\{\g_i\}$ as a permutation.
Because $\sigma^{2}=\mathrm{id}$, this permutation is an involution, and an involution has only two kinds of orbits: fixed points and two-element orbits.
Accordingly, $\sigma$ either fixes an ideal, $\sigma(\g_i)=\g_i$, or swaps two of them, $\sigma(\g_\alpha)=\g_\beta$ and $\sigma(\g_\beta)=\g_\alpha$; we write $m$ for the number of fixed ideals and $n$ for the number of swapped pairs.
The real points of a fixed ideal form a real form of that simple ideal: $\mathfrak{su}(2)$, $\mathfrak{sl}(2,\mathbb{R})$, $\mathfrak{su}(1,1)$, etc.
The real points of a swapped pair form a complex simple algebra regarded as real, of twice the complex dimension; the smallest is $\mathfrak{sl}(2,\mathbb{C})_\mathbb{R}\cong\mathfrak{so}(1,3)$.
With this decomposition of $\sigma$ in hand, we can state the attainable signatures.

\begin{theorem}[Attainable real signatures]\label{thm:signature}
Let $\gC$ be semisimple, with $m$ $\sigma$-fixed simple ideals of complex dimensions $d_i$ and $n$ swapped pairs of complex dimensions $e_\alpha$, so that $d=\sum_i d_i+2\sum_\alpha e_\alpha$.
Every nondegenerate $\sigma$-compatible solution $h$ of Eq.~\eqref{eq:master}, of signature $(p,q)$, falls into one of three cases:
\begin{itemize}
\item[(a)] $m\ge1$: $\min(p,q)\ge\max_i d_i\ge3$, whatever the couplings;
\item[(b)] $m=0$, $n=1$: $(p,q)=(2\tilde p,2\tilde q)$ is even--even, and every such signature occurs;
\item[(c)] $m=0$, $n\ge2$: every signature with $p+q=d$ occurs.
\end{itemize}
\end{theorem}

We sketch the three proofs here and give them in full in Appendix~\ref{app:B}.

\emph{Case (a): a fixed ideal forces a large Witt index.}
For a nondegenerate form of signature $(p,q)$, the largest dimension of a subspace on which the form vanishes identically is $\min(p,q)$; this number is called the Witt index.
For instance, in Minkowski space, a null line is such a subspace, and no null plane exists.
A fixed ideal carries $h^{(ii)}=0$, so $h$ vanishes identically on its real points, a subspace of dimension $d_i$, whatever the couplings to the rest.
The smallest complex simple Lie algebra is $\mathfrak{sl}(2,\mathbb{C})$, of dimension $3$, so $d_i\ge3$, and the Witt index bounds the signature: $\min(p,q)\ge d_i\ge3$.
Nondegenerate solutions in this case exist---$\mathfrak{su}(2)\oplus\mathfrak{su}(2)$ with an invertible coupling, of signature $(3,3)$, appears below---but none is Lorentzian or Euclidean, due to the lower bound.

\emph{Case (b): a single pair gives even--even signatures.}
With one swapped pair, the only block allowed by Corollary~\ref{cor:semisimple} is the coupling between the two ideals of the pair, and compatibility forces it to be a Hermitian matrix $B$.
The form on the real points is $2x^\dagger\bar{B}x$, with $\bar B=B^{T}$ again Hermitian; each eigenvalue of $\bar B$ contributes two real directions of its sign, so $(p,q)=(2\tilde p,2\tilde q)$, and choosing the eigenvalues realizes every even--even signature.
The negative-definite choice on the smallest pair, $\mathfrak{sl}(2,\mathbb{C})_\mathbb{R}\cong\mathfrak{so}(1,3)$, is the Euclidean saddle of Ref.~\cite{Liao:2025yfb}, with $h$ of signature $(0,6)$.

\emph{Case (c): two pairs attain every signature.}
Between two distinct pairs, the cross-couplings assemble into an arbitrary real matrix.
Take one negative direction from each pair, with the diagonal form $\mathrm{diag}(-2,-2)$, and turn on a cross-coupling of size $3$ between them: the block becomes $\left(\begin{smallmatrix}-2&3\\3&-2\end{smallmatrix}\right)$, with eigenvalues $\{1,-5\}$.
Exactly one direction has changed sign.
This adjustment has odd parity and is not available within a single pair, where case (b) forces even--even.
Starting from an even--even choice of the couplings and applying one such sign change covers every $(p,q)$.
The smallest pair has real dimension $2\times3=6$, so two pairs require $d\ge12$.
Hence below twelve dimensions the semisimple sector attains only the signatures of cases (a) and (b), while from twelve on it attains every signature.

Below twelve dimensions the complete list is short (Appendix~\ref{app:B}), and it is built from few blocks.
The smallest complex simple Lie algebra is $\mathfrak{sl}(2,\mathbb{C})$, of dimension $3$, with real forms $\mathfrak{su}(2)$ and $\mathfrak{sl}(2,\mathbb{R})\cong\mathfrak{su}(1,1)$; the next complex simple algebras have dimensions $8$ and $10$; the smallest swapped pair is $\mathfrak{so}(1,3)$, of real dimension $6$, and the next has real dimension $16$.
A real semisimple algebra of dimension $d\le11$ is therefore a sum of summands of real dimensions $3$, $8$, $10$ (fixed ideals) and $6$ (one pair).
Nondegenerate solutions exist only at $d=6$ and $d=9$.
At $d=6$, the algebra is either one pair---case (b), with the $(0,6)$ solution above---or a sum of two three-dimensional ideals, such as $\mathfrak{su}(2)\oplus\mathfrak{su}(2)$, whose invertible couplings give signature $(3,3)$.
At $d=9$, the algebra is a sum of three three-dimensional ideals, or of one such ideal and one pair; the simplest solution is $\mathfrak{su}(2)\oplus\mathfrak{su}(2)\oplus\mathfrak{su}(2)$ with all three couplings the identity, of signature $(3,6)$.
At $d=3,8,10$, the algebra is a single fixed ideal, and Theorem~\ref{thm:nogo} forces $h=0$; at $d=11$, the only decomposition is $3+8$, and $h$ has rank at most $6$, so every solution is degenerate.
No semisimple solution exists below $d=6$, consistent with the survey of Ref.~\cite{Chatzistavrakidis:2011su}, and none of the solutions at $d=6,9$ is Lorentzian.

The threshold at twelve is sharp.
At $d=12$, case (c) yields purely semisimple Lorentzian solutions: $\mathfrak{so}(1,3)\oplus\mathfrak{so}(1,3)$ with signature $(1,11)$, the timelike direction a cross-coupled combination of a boost generator from each factor (Appendix~\ref{app:D}).
Ten, the critical dimension of the superstring, lies strictly below this threshold (Fig.~\ref{fig:flow}, Table~\ref{tab:sectors}).

\subsection{Lorentzian signature requires the radical}\label{sec:lorentz}
This subsection applies the physical inputs to the classification.
The inputs are two: the signature of $h$ must be Lorentzian, $(p,q)=(1,d-1)$, and in the ten-dimensional model the algebra dimension satisfies $d\le k=10$.
We state the results for every $d\le11$, since the analysis is the same for all $d$ below the twelve-dimensional threshold of Theorem~\ref{thm:signature}.
Since $h_{ab}=C^\mu{}_a\,\eta_{\mu\nu}\,C^\nu{}_b$, and since $C$ is injective when $h$ is nondegenerate, $h$ is the spacetime bilinear form $\eta$ evaluated along the configuration: for $v\in\gR$, $h(v,v)=(Cv)^\mu\eta_{\mu\nu}(Cv)^\nu$ is the $\eta$-norm of the target-space direction $Cv$.
We can therefore read the timelike direction off $h$ directly, and we call $v$ timelike if $h(v,v)>0$, spacelike if $h(v,v)<0$, and null if $h(v,v)=0$.
A Lorentzian solution has a single positive direction: the time of the emergent background.
\begin{theorem}[Lorentzian solutions require a nonzero radical]\label{thm:main}
Let $\gR$ be a real Lie algebra of dimension $d\le11$, and let $h$ be a nondegenerate solution of Eq.~\eqref{eq:master} of Lorentzian signature $(1,d-1)$.
Then the radical of $\gR$ is nonzero.
In particular, every Lie-algebraic background of the ten-dimensional IKKT model that contains a timelike direction has $\rad\neq0$; Appendix~\ref{app:C} extends this to degenerate $h$.
\end{theorem}
\emph{Proof.}
If $\rad=0$, then $\gR$ is semisimple and Theorem~\ref{thm:signature} applies to $h$.
Case (a) would give $\min(p,q)\ge3$, but a Lorentzian signature has $\min(1,d-1)=1$; case (b) would make $p$ even, but $p=1$; case (c) would require $d\ge12$.
No case admits $(1,d-1)$ with $d\le11$, so $\rad\neq0$. \hfill$\square$

Theorem~\ref{thm:main} states that the radical is nonzero; it does not yet locate the timelike direction inside the algebra.
The next corollary locates it.

\begin{corollary}[Semisimple sector and time in Lorentzian solutions]\label{cor:radical}
Let $\gR$ be a real Lie algebra of dimension $d\le11$, and let $h$ be a nondegenerate solution of Eq.~\eqref{eq:master} of Lorentzian signature.
Then:
\begin{itemize}
\item[(i)] $\gss$ is either $0$ or the single swapped pair $\mathfrak{sl}(2,\mathbb{C})_\mathbb{R}\cong\mathfrak{so}(1,3)$;
\item[(ii)] the restriction of $h$ to $\gss$ is negative definite: every nonzero vector of $\gss$ is spacelike, and every timelike or null vector of $\gR$ has a nonzero component in $\rad$.
\end{itemize}
\end{corollary}

Appendix~\ref{app:B} gives the detailed proof; here we sketch the idea, statement by statement.

The first step is that $h|_{\gss}$ solves Eq.~\eqref{eq:master} on $\gss$ by itself, whatever the brackets between $\gss$ and $\rad$.
In general the Levi decomposition is semidirect, $[\gss,\rad]\subseteq\rad$, and these brackets need not vanish, so the different blocks of $h$---the block on $\gss$, the block on $\rad$, and the couplings between the two---could enter Eq.~\eqref{eq:master} together.
For $x\in\gss$, however, the contribution of the $\gss$ block is a sum of double brackets of elements of $\gss$ and lies in $\gss$, because $\gss$ is closed under the bracket, $[\gss,\gss]\subseteq\gss$, while the contributions of the other two blocks contain a generator of $\rad$ inside the brackets and lie in $\rad$, because $\rad$ is an ideal, $[\rad,\g]\subseteq\rad$.
Since $\gss$ and $\rad$ are linearly independent, the two parts vanish separately, and the $\gss$ part is Eq.~\eqref{eq:master} on $\gss$ for $h|_{\gss}$.

For statement (i): Theorem~\ref{thm:signature} now applies to $h|_{\gss}$.
By case (a) of Theorem~\ref{thm:signature}, a $\sigma$-fixed simple ideal would give a subspace of dimension at least $3$ on which $h$ vanishes, exceeding the Witt index $\min(1,d-1)=1$ of a Lorentzian form.
By case (c) of Theorem~\ref{thm:signature}, two swapped pairs would require $d\ge12$; and a single pair built on any simple algebra larger than $\mathfrak{sl}(2,\mathbb{C})$ has real dimension at least $16$.
The only possibilities within $d\le11$ are $\gss=0$ and the single pair $\mathfrak{so}(1,3)$.

For statement (ii): by case (b) of Theorem~\ref{thm:signature}, the single pair contributes an even number of positive directions to $h|_{\gss}$.
A restriction of $h$ has positive index at most $p=1$, and an even number that is at most $1$ is $0$.
Hence $h|_{\gss}\preceq0$.
The kernel of $h|_{\gss}$ also has even dimension, by case (b), and it is a subspace on which $h$ vanishes identically, so its dimension cannot exceed the Witt index $1$ of $h$; even and at most $1$ again forces $0$.
Hence $h|_{\gss}$ is negative definite: every nonzero vector of $\gss$ is spacelike.
A timelike or null vector $v$ has $h(v,v)\ge0$ and therefore cannot lie in $\gss$: its component in $\rad$ is nonzero.

The argument above uses only the $\gss$ part of Eq.~\eqref{eq:master}; the remaining parts are additional conditions.
For $c\in\gss$, the $\rad$ part constrains the couplings between $\gss$ and $\rad$; for $c\in\rad$, the equation involves $h|_{\gss}$ through the brackets between $\gss$ and $\rad$; Appendix~\ref{app:B} writes these remaining parts out.
These conditions decide whether a given algebra with $[\gss,\rad]\neq0$ admits nondegenerate Lorentzian solutions at all, a case-by-case question that Sec.~\ref{sec:radical} takes up; statements (i) and (ii) hold for every algebra.

Theorem~\ref{thm:main} and Corollary~\ref{cor:radical} assume $h$ nondegenerate.
In the ten-dimensional model a configuration can span fewer than ten directions, and $h$ can then be degenerate; the pullback bounds of Sec.~\ref{sec:setup} extend the analysis to this case (Appendix~\ref{app:C}), and the conclusion holds for either target signature.
With $\eta$ Euclidean or Lorentzian and $d\le k=10$, the semisimple part of every Lie-algebraic configuration is $0$ or the single pair $\mathfrak{sl}(2,\mathbb{C})_\mathbb{R}\cong\mathfrak{so}(1,3)$, and $h$ is negative definite on it, of signature $(0,6)$.
In particular, a purely semisimple configuration of the ten-dimensional model, of either signature, is the $(0,6)$ saddle of Ref.~\cite{Liao:2025yfb}: it occupies only spacelike directions of the target space, and no semisimple configuration reaches the timelike direction.

\section{Solutions with a radical: direct and semidirect sums}\label{sec:radical}

By Corollary~\ref{cor:radical}, the timelike direction of a Lorentzian solution has a component in the radical.
This section constructs explicit solutions and determines how the two parts of the Levi decomposition can combine, either as a direct sum, $[\gss,\rad]=0$, or as a genuine semidirect sum, $[\gss,\rad]\neq0$.
One example at each of $d=9$, $10$, and $11$ demonstrates the possibilities.

\subsection{Direct sums}

In a direct sum, generators of distinct summands commute, so every double bracket in Eq.~\eqref{eq:master} stays inside one summand.
The equation decouples into the blocks, and the couplings between summands are unconstrained; the radical block must therefore solve Eq.~\eqref{eq:master} on the radical alone.
For a Lorentzian solution with a nontrivial semisimple part, that part is $\mathfrak{so}(1,3)$ (Corollary~\ref{cor:radical}), so for $d\le10$ the radical has dimension at most four, and the classification tables of all Lie algebras of dimension up to five~\cite{Chatzistavrakidis:2011su} list every candidate radical.
If $[\rad,\rad]$ is central (a $2$-step nilpotent ideal), the double commutator in Eq.~\eqref{eq:master} vanishes identically, and \emph{any} $h$ of \emph{any} signature solves the constraint there, so time can sit directly on a nilpotent direction.
The Heisenberg algebra $\mathfrak{h}_3$, $[X,Y]=Z$, with $h=\mathrm{diag}(+1,-1,-1)$, is the minimal Lorentzian solution, of signature $(1,2)$; it is the Moyal plane with noncommuting time, $[X^0,X^1]\propto\mathbb{I}$.
Nilpotent algebras of this type underlie the solutions found in the survey of Ref.~\cite{Chatzistavrakidis:2011su} and the nilmanifold backgrounds of Ref.~\cite{ChatzistavrakidisJonke2012}.
Sectors combine by direct sum; $\mathfrak{so}(1,3)\oplus\mathfrak{h}_3$ with $h=(0,6)\oplus(1,2)$ is a $(1,8)$ solution, and the block-diagonal $h$ splits the matrix Laplacian into a metric product of the Euclidean geometry of Ref.~\cite{Liao:2025yfb} with a Moyal plane carrying time.

Solvable algebras that are not nilpotent behave differently. Their internal solutions are degenerate, with null directions that the couplings then cure.
The smallest is $\mathfrak{aff}(1)=\{X,Y:[X,Y]=Y\}$, where Eq.~\eqref{eq:master} forces $h_{XX}=h_{XY}=0$: the generator $X$ that rescales $Y$ must be null and orthogonal to everything else, so $\mathfrak{aff}(1)$ alone has only degenerate solutions.
A coupling to any other direction cures this.
Already $\mathfrak{aff}(1)\oplus\mathbb{R}$, with $h_{XZ}=1$, $h_{YY}=-1$, and all other entries zero, is a nondegenerate $(1,2)$ solution at $d=3$: the null generator $X$ and the commuting generator $Z$ form a hyperbolic pair, and the timelike direction is $X+Z$.
The same mechanism explains the expanding solution of Refs.~\cite{Kim:2011ts,Kim:2012mw}, whose algebra $[X,Y_i]=Y_i$, $[Y_i,Y_j]=0$ likewise forces $h_{XX}=h_{XY_i}=0$: with $h$ equal to the flat metric, that background solves only the regularized equation, in which the spatial multiplier supplies the missing term (Appendix~\ref{app:E}).
Coupling the null generator $X$ of $\mathfrak{aff}(1)$ to $\mathfrak{so}(1,3)$ gives a nondegenerate $(1,7)$ solution at $d=8$, and appending two commuting directions gives a full-rank $(1,9)$ background of the physical model at $d=10$, on $\mathfrak{so}(1,3)\oplus\mathfrak{aff}(1)\oplus\mathbb{R}^{2}$ (Appendix~\ref{app:D}).
Lorentzian signature requires a nonzero radical, not nilpotency.

The example at $d=9$ is built on $\mathfrak{e}(2)=\{J,P_1,P_2:[J,P_1]=P_2,\ [J,P_2]=-P_1\}$, the algebra of rotations and translations of the plane, where Eq.~\eqref{eq:master} forces $h_{JJ}=h_{JP_i}=0$ and leaves the $P$ block free, so $J$ is null internally, exactly like $X$ above; the two-dimensional Poincar\'e algebra $\mathfrak{e}(1,1)$ behaves identically.
This null rotation generator is the algebraic form of the ``split noncommutativity'' of Ref.~\cite{Steinacker:2011wb}: the propagating fuzzy cylinder of that work places the generator $J$ of $\mathfrak{e}(2)$ along a light-like direction of the target space, and its propagating plane wave places the rotation generator on a light-cone coordinate $X^{-}$ paired with $X^{+}$ through $[X^{+},X^{-}]\propto\mathbb{I}$, a five-dimensional solvable algebra whose null generator is cured by a hyperbolic coupling.
The observation of Ref.~\cite{Steinacker:2011wb} that these compactified solutions exist only for Minkowski signature is, in the present language, the statement that a null generator needs a null direction of $\eta$.
On $\mathfrak{so}(1,3)\oplus\mathfrak{e}(2)$, order the basis as (the six generators of $\mathfrak{so}(1,3)$; $J$; $P_1,P_2$) and take
\begin{equation}\label{eq:e218}
h=\begin{pmatrix} h_{0} & c & 0\\ c^{T} & 0 & 0\\ 0 & 0 & -\mathbb{I}_2 \end{pmatrix},
\end{equation}
where $h_{0}$ is the negative-definite $(0,6)$ solution on $\mathfrak{so}(1,3)$ from case (b) of Theorem~\ref{thm:signature} and $c\in\mathbb{R}^{6}$ is the coupling of $J$.
Completing the square in $J$ brings $h$ to $h_{0}\oplus(-c^{T}h_{0}^{-1}c)\oplus(-\mathbb{I}_2)$ by a change of basis, and $-c^{T}h_{0}^{-1}c>0$ because $h_{0}$ is negative definite.
Hence for every $c\neq0$ the solution is nondegenerate of signature $(1,8)$, with the timelike direction carried by the null generator $J$---the same odd-parity adjustment as in Theorem~\ref{thm:signature}(c), now operating between $\gss$ and $\rad$.

\subsection{Semidirect sums}

We now turn to genuine semidirect sums, $[\gss,\rad]\neq0$.
For a Lorentzian solution the semisimple part is $\mathfrak{so}(1,3)$ (Corollary~\ref{cor:radical}), and the brackets $[\gss,\rad]\subseteq\rad$ make the radical a representation space of $\mathfrak{so}(1,3)$.
The Jacobi identity for $x\in\gss$ and $u,v\in\rad$ reads $[x,[u,v]]=[[x,u],v]+[u,[x,v]]$; it states that transforming the bracket of two radical elements equals bracketing the transformed elements, so the bracket of the radical is a map compatible with the representation---an equivariant map.
Representation theory then fixes the possibilities.

\begin{theorem}[Semidirect structure]\label{thm:semi}
Let $h$ be a nondegenerate solution of Eq.~\eqref{eq:master} of Lorentzian signature $(1,d-1)$, $d\le11$, with $[\gss,\rad]\neq0$.
Then:
\begin{itemize}
\item[(i)] such a solution requires $d\ge10$; that is, for $d\le9$ every nondegenerate Lorentzian solution has $[\gss,\rad]=0$;
\item[(ii)] if $d=10$, the radical is abelian, $\rad\cong\mathbb{R}^{4}$, carrying the vector or the Weyl-spinor representation of $\mathfrak{so}(1,3)$, and nondegenerate solutions exist only in the spinor case;
\item[(iii)] at $d=11$ a nonabelian radical first becomes possible, and it is realized, for example with the timelike direction along a central element.
\end{itemize}
\end{theorem}

Appendix~\ref{app:G} gives the proof; the ideas are the following.

The smallest nontrivial real representation of $\mathfrak{so}(1,3)$ has dimension four.
Complex conjugation exchanges the two simple ideals of the complexification---the swapped pair of Sec.~\ref{sec:real}---so every real representation pairs the two chiral factors, and nothing nontrivial fits below the vector and the Weyl spinor.
For $d\le9$ the radical has dimension at most three and must therefore be a trivial representation, which is statement (i).

Explicitly, the vector representation is the familiar action on a four-vector $(x^0,x^1,x^2,x^3)$, the rotations mixing the $x^i$ and the boosts mixing $x^0$ with the $x^i$; the Weyl-spinor representation acts on a complex doublet through the Pauli matrices $\sigma_i$, with $[J_i,\cdot]=-\tfrac{i}{2}\sigma_i$ and $[K_i,\cdot]=-\tfrac{1}{2}\sigma_i$, displayed in Eq.~\eqref{eq:weylact} below.

At $d=10$ the radical is one of these two four-dimensional representations, and its own bracket must vanish.
The bracket is antisymmetric, $[u,v]=-[v,u]$, so it is an equivariant map from the antisymmetric pairs $u\wedge v$ of radical elements back to the radical; decomposing the antisymmetric pairs into irreducible representations, none matches the radical itself, and an equivariant map between distinct irreducible representations vanishes (Appendix~\ref{app:G}).
The equations for $c\in\rad$ then involve $h|_{\gss}$ through the representation, and they separate the two cases.
On the vector representation they exclude every nonzero $h|_{\gss}$, so the Poincar\'e algebra admits no nondegenerate solution at all.
On the spinor representation they impose nothing, because every solution $h|_{\gss}$ couples only the two chiral factors (Corollary~\ref{cor:semisimple}) and each factor acts as zero on one half of the representation.
This is statement (ii); the following two examples realize (ii) and (iii).

\subsubsection{Spinor translations at $d=10$}

This example sits at the critical dimension of superstring theory.
Take four commuting generators $Q_1,\dots,Q_4$ and let $\mathfrak{so}(1,3)$, with rotations $J_i$ and boosts $K_i$, act by the Weyl-spinor representation.
On the complex doublet $\mathcal{Q}=(Q_1+iQ_3,\;Q_2+iQ_4)^{T}$, the brackets are
\begin{equation}\label{eq:weylact}
[J_i,\mathcal{Q}]=-\tfrac{i}{2}\,\sigma_i\,\mathcal{Q},\qquad
[K_i,\mathcal{Q}]=-\tfrac{1}{2}\,\sigma_i\,\mathcal{Q},
\end{equation}
so that $[X,Q_a]=\rho(X)^{b}{}_{a}Q_b$ with $\rho$ the realified spinor matrices.
The only difference from the Poincar\'e algebra is the representation of the translations, a spinor instead of a vector.
The algebra is bosonic, since no Grassmann numbers appear.
Faithful unitary representations exist, by the induced-representation construction for a semidirect sum with commuting translations~\cite{Mackey1952}.
Then
\begin{equation}\label{eq:weylsol}
h=\mathrm{diag}(\underbrace{-1,\dots,-1}_{6},\;+1,-1,-1,-1)
\end{equation}
solves Eq.~\eqref{eq:master}.
For $c\in\mathfrak{so}(1,3)$, the $Q$-terms vanish because the translations commute, and the remaining condition is the standalone $\mathfrak{so}(1,3)$ equation, solved by $h\propto\delta$ (Sec.~\ref{sec:complex}).
For $c=Q_b$, the condition is $\sum_i\big(\rho(J_i)^2+\rho(K_i)^2\big)Q_b=0$, and Eq.~\eqref{eq:weylact} gives $\rho(K_i)=-i\,\rho(J_i)$, so the two sums cancel term by term.
The solution is nondegenerate of signature $(1,9)$; it is a genuinely semidirect Lorentzian background of the ten-dimensional model, with every Lorentz generator acting on the translations and the timelike direction on the spinorial translation $Q_1$.
In the induced representation the $Q_a$ act as multiplication operators, so the timelike matrix has a continuous spectrum, as time should.

\subsubsection{A nonabelian radical at $d=11$}

The spinor representation carries an invariant antisymmetric form, $\varepsilon(\mathcal{Q},\mathcal{Q}')=\mathcal{Q}_1\mathcal{Q}'_2-\mathcal{Q}_2\mathcal{Q}'_1$, preserved by Eq.~\eqref{eq:weylact}; its real part $\omega=\mathrm{Re}\,\varepsilon$ is a real antisymmetric form on the $Q_a$.
Extend the algebra by one central element $Z$,
\begin{equation}\label{eq:heis}
[Q_a,Q_b]=\omega_{ab}\,Z,\qquad [Z,\cdot\,]=0.
\end{equation}
The radical becomes a five-dimensional Heisenberg algebra $\mathfrak{h}_5$ on which $\mathfrak{so}(1,3)$ still acts.
The only new condition, the $Z$-component of the equations for $c=x\in\mathfrak{so}(1,3)$, reads $\mathrm{tr}\big(\omega\,\rho(x)\,h_{QQ}\big)=0$ for all $x$.
The block $h_{QQ}=-\mathbb{I}_4$ satisfies it: the condition becomes $\sum_a\omega(Q_a,\rho(x)Q_a)=0$, the complex bilinearity of $\varepsilon$ gives $\omega(iv,\rho(x)\,iv)=-\omega(v,\rho(x)v)$, and the real basis pairs each $v$ with $iv$, so the sum cancels (Appendix~\ref{app:G}).
Hence
\begin{equation}\label{eq:heissol}
h=\mathrm{diag}(\underbrace{-1,\dots,-1}_{6},\;\underbrace{-1,-1,-1,-1}_{4},\;+1)
\end{equation}
solves Eq.~\eqref{eq:master}, nondegenerate of signature $(1,10)$.
Here the timelike direction is the center: the noncommutativity of the spinorial translations generates time.
Two remarks make this precise.
First, since $Z$ is central, no entry of $h$ in its row enters any equation, so the timelike direction can equally be tilted toward the $Q_a$; for instance $h_{Q_1Z}=2$ with all diagonal entries $-1$ is again a $(1,10)$ solution.
Second, an irreducible unitary representation---the Schr\"odinger representation of $\mathfrak{h}_5$, with the Lorentz sector acting by quadratic operators in the oscillator representation~\cite{Folland1989}---maps the central generator $Z$ to a multiple of the identity, and a constant matrix direction is trivial by the translation symmetry $A^\mu\to A^\mu+c^\mu\mathbb{I}$ of the model.
The timelike matrix of Eq.~\eqref{eq:heissol} therefore acquires a spectrum only in a reducible representation---a direct integral of Schr\"odinger representations over the central parameter, in which the noncommutativity scale of the spinorial translations plays the role of the time coordinate, the same mechanism that gives the commuting spatial matrices of the expanding-universe solutions~\cite{Kim:2011ts} their spread spectrum---or after the tilt just described.

\section{Discussion}\label{sec:disc}
The classification answers the two questions of the introduction: which solutions the Lie-algebraic ansatz admits for either target signature, and where the timelike direction resides.
Within the Lie-algebraic class---every algebra of dimension $d\le11$, hence every background of the ten-dimensional model---the timelike direction of a nondegenerate Lorentzian solution lies in the solvable radical; the semisimple sector supports Euclidean, split, or degenerate geometry only.
When the semisimple part acts on the radical, the physical dimension is the threshold.
At $d=10$ the radical must be abelian, and only spinor-valued translations yield semidirect backgrounds, while the Poincar\'e algebra yields none; a nonabelian radical first appears at $d=11$ (Sec.~\ref{sec:radical}).

How to extract an effective geometry from a given background is, in general, an open question.
Here the representation, kept implicit so far, becomes central: one reads the geometry off the operators acting on the Hilbert space of a chosen unitary representation.
One proposal is the semiclassical picture of Ref.~\cite{Steinacker:2010rh}: the commutators of the background define a Poisson tensor $\theta$, and fluctuations $A^\mu\to A^\mu+\delta A^\mu$ propagate in an effective metric $G=e^{-\sigma}\,\theta\,g\,\theta^{T}$, built pointwise from $\theta$ and the pullback $g$ of $\eta$ to the semiclassical brane; this is the picture adopted in Ref.~\cite{Liao:2025yfb} for the $(0,6)$ solution, whose effective geometry is a four-dimensional Euclidean space.
The classification bounds every signature obtainable in this picture.
Each step from $\eta$ to $G$---$h=C^{T}\eta C$, the pullback $g$ of $h$, and the contraction with $\theta$---is a congruence $s\mapsto M^{T}sM$ by a real matrix.
A congruence cannot increase the number of positive or negative eigenvalues: if $M^{T}sM$ is positive definite on a subspace $W$, then $s$ is positive definite on $M(W)$.
Hence $p(G)\le p(h)\le p(\eta)$ pointwise, and likewise for the negative counts: an effective timelike direction requires a timelike direction of $h$.

The expanding-universe backgrounds of the Lorentzian model~\cite{Kim:2011ts,Kim:2012mw} have exactly the structure that Theorem~\ref{thm:main} requires: commutative space, $[X^i,X^j]=0$, acted on by time through $[X^0,X^i]\propto X^i$, is a solvable algebra.
Both branches of that construction, however, depend on the infrared regulators.
The cutoff terms add a multiplier $\lambda_{(\mu)}A^\mu$ to the equation of motion.
On the $\mathfrak{su}(2)$ and $\mathfrak{su}(1,1)$ branches, with $h$ proportional to the inverse Killing form, the double bracket equals the adjoint Casimir eigenvalue times $A^\mu$, so a solution exists exactly when the multiplier matches that eigenvalue; on the expanding branch, the spatial multiplier supplies the term that the undeformed equation would force to vanish by making the time generator null (Appendix~\ref{app:E}).
The split-noncommutativity solutions of Ref.~\cite{Steinacker:2011wb} realize the null-generator mechanism of Sec.~\ref{sec:radical} without regulators; their fuzzy-sphere factors, multiplied by exponentials of noncommutative coordinates, do not close into a finite-dimensional Lie algebra, which is how they evade Theorem~\ref{thm:nogo}.
The present results are kinematical; they complement dynamical signature-selection mechanisms such as the Lorentz-invariant mass deformation of Refs.~\cite{Nishimura:2022alt,Hirasawa2023reg}.

The BFSS matrix model~\cite{Banks:1996vh} differs at exactly this point.
It is the $0{+}1$-dimensional reduction: nine matrices $X^i(t)$ in an external, commutative time, which enters as the derivation $d/dt$---the canonical pair $[\partial_t,t]=1$, an infinite-dimensional realization of the solvable structure required above; compactifying one direction of the IKKT model produces BFSS with the compact direction acting as a covariant derivative~\cite{Taylor:1996ik}.
Time-independent BFSS configurations satisfy $\sum_j[X^j,[X^j,X^i]]=0$, which is Eq.~\eqref{eq:EOM} with Euclidean $\eta$ and $k=9$, so Theorem~\ref{thm:nogo} applies: the fuzzy sphere is not a time-independent BFSS solution---it is a vacuum of the mass-deformed BMN model~\cite{Berenstein:2002jq}, whose $\mu$-terms supply the Casimir eigenvalue---while the flat membrane $[X^1,X^2]\propto\mathbb{I}$, which is nilpotent, is one.
Conversely, every solution of Eq.~\eqref{eq:master} with $d\le9$ and definite $h$ is a time-independent BFSS background.

The construction also seeds time-dependent solutions: promoting the coefficients to functions of time turns the BFSS equation of motion into an ordinary differential equation for $C(t)$ on a fixed algebra---equivalently, a time-dependent $h(t)=C(t)^{T}\eta\,C(t)$---with the time-independent solutions as fixed points.
The collapsing fuzzy sphere, $\mathfrak{su}(2)$ with a time-dependent overall scale, is the simplest example, consistent with Theorem~\ref{thm:nogo}: simple algebras, excluded in the time-independent problem, reappear as time-dependent solutions.
The minimal Heisenberg solution $[A^0,A^1]\propto\mathbb{I}$ corresponds in BFSS variables to a D0-brane in uniform motion, $[D_t,X^1]=v\,\mathbb{I}$.
BFSS contains solvable time by definition; in the IKKT model it must arise as a property of the solution, and by Theorem~\ref{thm:main} there is no other option within the class.

The covariant formulations use the structures identified here in infinite-dimensional form.
In the covariant-derivative interpretation~\cite{HanadaKawaiKimura2006,Tsuchiya2024reg} and in the derivation of general relativity from the model in Ref.~\cite{HoKawaiSteinacker2026}, the matrices act as differential operators: the vacuum is the abelian algebra $A_\mu\sim i\partial_\mu$, and derivations acting on commutative functions carry the geometry, with brackets of Heisenberg type $[\partial_\mu,x^\nu]=\delta_\mu^\nu$.
The covariant quantum spacetimes of Ref.~\cite{SperlingSteinacker2019} start from the semisimple $\mathfrak{so}(4,2)$ and are consistent with Theorem~\ref{thm:nogo} in the two available ways: the generators do not close into a finite-dimensional Lie algebra (Snyder-type~\cite{Snyder:1946qz}, outside the present class), and the construction lives in the mass-deformed model.
Determining which subalgebras of differential operators solve Eq.~\eqref{eq:master} is an open problem that would connect the two settings.

The following problems remain open.

(i) In case (a) of Theorem~\ref{thm:signature}, the fixed ideals force the lower bound $\min(p,q)\ge d_i$, but which signatures above this bound occur, for a given set of ideals and couplings, remains undetermined.

(ii) Solvable algebras of dimension six and higher are not classified, and Appendix~\ref{app:F} (Table~\ref{tab:landscape}) collects the known solutions with $d\le11$.
Even in dimension $\le5$, where the algebras are classified, the survey of Ref.~\cite{Chatzistavrakidis:2011su} tests only the flat metric in the tabulated basis; solving Eq.~\eqref{eq:master} for general $h$ on each classified algebra is a finite linear-algebra task that has not been carried out, and the $\mathfrak{aff}(1)\oplus\mathbb{R}$ example of Sec.~\ref{sec:radical} shows that it yields new Lorentzian solutions.

(iii) The semidirect analysis of Sec.~\ref{sec:radical} covers $d\le11$; a systematic treatment of the remaining conditions---the couplings and the equations for $c\in\rad$ on higher-dimensional representations, the full scaling family at $d=11$, and $d\ge12$---is open.

(iv) The extraction of the effective geometry beyond the semiclassical picture above is open, and with it the question of which effective signatures the solutions realize, rather than merely bound.
Part of this question is the choice of representation: central and abelian generators act as constants in irreducible representations and acquire a spectrum only in direct integrals (Sec.~\ref{sec:radical}).

(v) We have not analyzed the stability of the explicit solutions under fluctuations, or the supersymmetry they preserve, if any.

(vi) How does the classification change under mass and Myers deformations?
They add source terms to Eq.~\eqref{eq:master}, so the Casimir obstruction becomes a solvability condition (Appendix~\ref{app:E}) and semisimple algebras re-enter; which signatures the deformed equation then allows, and where time can reside, is open.
The question connects the classification to the polarized IKKT and BMN models~\cite{Hartnoll:2024csr,Komatsu:2024bop,Berenstein:2002jq}, where supersymmetric solutions are studied holographically.

(vii) Whether quantum fluctuations select the time-carrying solvable sectors, as the Lorentzian Monte Carlo results~\cite{Kim:2011cr,Anagnostopoulos:2022dak} suggest, requires a large-$N$ analysis.

\section*{Acknowledgments}
H.L. thanks Hikaru Kawai, Jun Nishimura, and Cheng-Tsung Wang for useful and insightful discussions.
H.L. is supported in part by the National Science and Technology Council (NSTC) grants 112-2112-M-002-024-MY3 and 112-2628-M-002-003-MY3.

\begin{figure}[tb]
\centering
\begin{tikzpicture}[every node/.style={font=\footnotesize}, node distance=4mm]
\node[draw, rounded corners, align=center] (g) at (0,2.1) {Lie-algebraic solution:\quad {$\g\cong\gss\ltimes\rad$}};
\node[draw, rounded corners, align=center] (ss) at (-1.9,0.9) {semisimple $\gss$};
\node[draw, rounded corners, align=center] (r) at (1.9,0.9) {radical $\rad$};
\node[align=center] (ssc) at (-1.9,-0.75) {fixed ideal: $\min(p,q)\ge3$\\ one pair: $(2\tilde p,2\tilde q)$\\ $\ge2$ pairs ($d\ge12$): all $(p,q)$\\ \emph{rigid below 12}};
\node[align=center] (rc) at (1.9,-0.75) {nilpotent: any $(p,q)$\\ null $+$ coupling: $(1,q)$\\ \emph{carries time}};
\node[draw, rounded corners, fill=gray!12, align=center] (concl) at (0,-2.45) {$d\le11$ (IKKT: $d\le10$): Lorentzian $\Rightarrow$ $\rad\neq0$,\ {$h|_{\gss}\prec0$}\\ time points into $\rad$\quad (sharp at $d=12$)};
\draw[->] (g) -- (ss); \draw[->] (g) -- (r);
\draw[->] (ss) -- (ssc); \draw[->] (r) -- (rc);
\draw[->] (ssc) -- (concl); \draw[->] (rc) -- (concl);
\end{tikzpicture}
\caption{\label{fig:flow}Signatures of Lie-algebraic solutions, organized by the Levi decomposition.
By Theorem~\ref{thm:signature}, below twelve dimensions the semisimple sector attains only constrained signatures, and from twelve on it attains all of them; for $d\le11$ the timelike direction of a nondegenerate Lorentzian solution lies in the radical, and $h$ is negative definite on the semisimple part (Theorem~\ref{thm:main}, Corollary~\ref{cor:radical}).}
\end{figure}
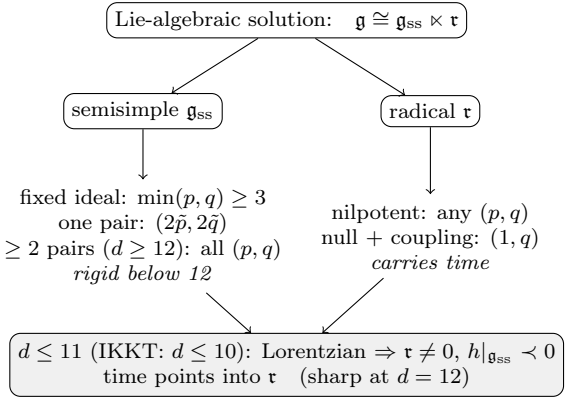

\begin{table*}[tb]
\caption{\label{tab:sectors}Signature constraints on the sectors of a Lie-algebraic solution with nondegenerate $h$.
The first three rows are the three cases of Theorem~\ref{thm:signature}; the last two are the radical mechanisms of Sec.~\ref{sec:radical}.
Below $d=12$, only the non-semisimple rows can supply a lone timelike direction (Corollary~\ref{cor:radical}).}
\begin{ruledtabular}
\begin{tabular}{lllllc}
Sector & Case & Structure & Signatures & Example & Time? \\
\hline
$\sigma$-fixed simple ideal & \ref{thm:signature}(a) & totally isotropic & $\min(p,q)\ge d_i\ge3$ & $\mathfrak{su}(2)^{\oplus2}$: $(3,3)$ & no \\
single swapped pair & \ref{thm:signature}(b) & Hermitian coupling $B$ & $(2\tilde p,2\tilde q)$, even--even & $\mathfrak{so}(1,3)$: $(0,6),(2,4)$ & no \\
coupled swapped pairs & \ref{thm:signature}(c) & free cross-couplings & all $(p,q)$; needs $d\ge12$ & $\mathfrak{so}(1,3)^{\oplus2}$: $(1,11)$ & if $d\ge12$ \\
$2$-step nilpotent radical & --- & $[\rad,\rad]$ central & unconstrained & $\mathfrak{h}_3$: $(1,2)$ & yes \\
solvable, null coupling & --- & hyperbolic pairing with $\gss$ & $(+,-)$ pairs & $\mathfrak{so}(1,3){\oplus}\mathfrak{aff}(1)[{\oplus}\mathbb{R}^2]$: $(1,7)[(1,9)]$ & yes \\
\end{tabular}
\end{ruledtabular}
\end{table*}

\appendix

\section{Proof of Theorem~\ref{thm:nogo}}\label{app:A}

This appendix proves Theorem~\ref{thm:nogo} in three steps: an explicit computation for $\mathfrak{su}(2)$, the two-spin form of the constraint for a general simple algebra, and a check of the required Casimir eigenvalue against the classification.

\subsection{Warm-up: $\mathfrak{su}(2)$}
With $[J_a,J_b]=\epsilon_{abc}J_c$ and $\sum_d\epsilon_{bcd}\epsilon_{ade}=\delta_{be}\delta_{ca}-\delta_{ba}\delta_{ce}$,
\begin{equation}\label{eq:app-su2}
  h_{ab}[J_a,[J_b,J_c]]=h_{ab}\,\epsilon_{bcd}\,\epsilon_{ade}\,J_e=\big(h-\mathrm{tr}(h)\,\mathbb{I}\big)_{ce}J_e .
\end{equation}
Hence $K_H=h-\mathrm{tr}(h)\,\mathbb{I}$, and $K_H=0$ forces $h=\mathrm{tr}(h)\,\mathbb{I}$; the trace of this relation gives $\mathrm{tr}\,h=3\,\mathrm{tr}\,h$, so $h=0$.
Split $h=\tfrac13\mathrm{tr}(h)\,\mathbb{I}+h_0$ into a singlet and a traceless, spin-$2$ part: then $K_H=h_0-\tfrac23\mathrm{tr}(h)\,\mathbb{I}$, so the two parts are multiplied by the nonzero numbers $1$ and $-2$, and no nonzero $h$ can satisfy $K_H=0$.
The general proof below shows that every simple algebra behaves this way, with the two numbers replaced by $c_\g-\tfrac12C_\lambda$, one for each irreducible constituent of $h$.
(In this real basis $\sum_c\ad_c^2=-2\,\mathbb{I}$, so $c_\g=-2$ and $C_{\text{spin-}2}=-6$, and $c_\g-\tfrac12C_\lambda$ reproduces $-2$ and $1$.)

\subsection{Two-spin form of the constraint}
Fix the Killing form $\kappa(x,y)=\mathrm{tr}(\ad_x\ad_y)$ of the complex simple algebra $\g$ and choose a basis with $\kappa(T_a,T_b)=\delta_{ab}$, so that the structure constants $[T_a,T_b]=f_{abc}T_c$ are totally antisymmetric and $\ad_a^{T}=-\ad_a$.
Write $H$ for the matrix $(h_{ab})$, viewed as an operator on $\g$.
Total antisymmetry gives the rearrangement
\begin{equation}\label{eq:app-KH}
  K_H\equiv\sum_{a,b}h_{ab}\,\ad_a\ad_b=\sum_c \ad_c\,H\,\ad_c :
\end{equation}
both sides have the matrix elements $\sum_{a,b,c}h_{ab}f_{acx}f_{bcy}$ once $f_{cax}=-f_{acx}$ is used.
Now read $H$ as a two-particle state.
The algebra acts on $\mathrm{End}(\g)\cong\g\otimes\g$ by $H\mapsto\ad_xH-H\ad_x$: left multiplication by $\ad_x$ is the generator $\vec J_1$ of the first particle, and right multiplication by $-\ad_x$ is the generator $\vec J_2$ of the second.
In this language $\sum_c\ad_c^2=c_\g\,\mathbb{I}$ reads $\vec J_1^{\,2}=\vec J_2^{\,2}=c_\g$ (with $c_\g=1$ in the present normalization; we keep it symbolic since only ratios matter), and Eq.~\eqref{eq:app-KH} says
\begin{equation}\label{eq:app-JJ}
  K_H=-\,\vec J_1\!\cdot\!\vec J_2\,H,\qquad
  \vec J_1\!\cdot\!\vec J_2\,H\equiv-\sum_c\ad_c\,H\,\ad_c .
\end{equation}
The total Casimir of the two-particle system is $\Omega(H)\equiv\sum_c[\ad_c,[\ad_c,H]]=(\vec J_1+\vec J_2)^{2}H$, and the familiar identity $(\vec J_1+\vec J_2)^{2}=\vec J_1^{\,2}+\vec J_2^{\,2}+2\vec J_1\!\cdot\!\vec J_2$ reads $\Omega(H)=2c_\g H-2K_H$.
Hence
\begin{equation}\label{eq:app-casimir}
  K_H=0\iff \Omega(H)=2c_\g\,H .
\end{equation}
Symmetric $h$ live in $\Sym^2(\g)\subset\g\otimes\g$ ($K_H$ is symmetric whenever $H$ is, by $\ad_c^{T}=-\ad_c$), and $\Omega$ acts on an irreducible constituent $V_\lambda$ by the number $C_\lambda=\langle\lambda,\lambda+2\rho\rangle$.
On that constituent $K_H=\big(c_\g-\tfrac12C_\lambda\big)H$, exactly as in the warm-up, so nonzero solutions exist if and only if $\Sym^2(\g)$ contains a constituent with $C_\lambda=2c_\g$, that is, if and only if the ratio $2$ appears in Table~\ref{tab:casimir}.

\subsection{The eigenvalue $2c_\g$ never occurs}
It never does.
For the classical families the candidate ratios are $2\pm(\text{nonzero})$, together with $n/(n-2)$ for $\mathfrak{so}_n$---equal to $2$ only at $n=4$, which is not simple---and $n/(n+1)<1$ for $\mathfrak{sp}_{2n}$; for the exceptional algebras we inspect the listed rational values directly.
Constituents of dimension zero at small rank (the adjoint is absent from $\Sym^2\mathfrak{sl}_2$, the constituent with eigenvalue $2-\tfrac2n$ is absent from $\Sym^2\mathfrak{sl}_3$, and the fourth universal constituent vanishes for $\mathfrak{g}_2$) only remove candidates and cannot produce the value $2$.
Hence the only solution is $h=0$, proving Theorem~\ref{thm:nogo}.

Two features of the table are worth noting.
The trivial constituent (the inverse Killing form) has eigenvalue $0$: there $K_H$ is the adjoint Casimir itself, nonzero.
The top constituent $V_{2\theta}$ has $\langle2\theta,2\theta+2\rho\rangle=2c_\g+2\langle\theta,\theta\rangle>2c_\g$ and is the only one exceeding $2c_\g$.
The table follows the universal decomposition of $\Sym^2(\g)$~\cite{Vogel1999,Deligne1996,LandsbergManivel2006,Mkrtchyan2012}: with Vogel parameters $(\alpha,\beta,\gamma)$, normalized so $\alpha=-2$, and $t=\alpha+\beta+\gamma$, one has $\Sym^2(\g)=\mathbb{C}\oplus Y_2(\alpha)\oplus Y_2(\beta)\oplus Y_2(\gamma)$ with Casimir eigenvalues
\begin{equation}\label{eq:app-vogel}
  \{0,\;4t-2\alpha,\;4t-2\beta,\;4t-2\gamma\},\qquad 2c_\g=4t ;
\end{equation}
since $\beta,\gamma>0$ for every simple algebra, only $4t-2\alpha=4t+4$ exceeds $4t$, and no eigenvalue equals $4t$.
(The constituents of $\g\otimes\g$ with eigenvalue exactly $2c_\g$ lie in $\Lambda^2\g$; there the two spins can be orthogonal, which is why the symmetry of $h$ matters.)
Table~\ref{tab:casimir} is the classification-complete check of this statement; we have also confirmed Theorem~\ref{thm:nogo} by solving the linear system~\eqref{eq:master} directly for $\mathfrak{sl}_n$ ($n\le4$), $\mathfrak{so}_n$ ($5\le n\le7$), and $\mathfrak{sp}_{2n}$ ($n\le3$).

\begin{table}[tb]
\caption{\label{tab:casimir}Quadratic-Casimir eigenvalues, in units of the adjoint Casimir $c_\g$, on the irreducible constituents of $\Sym^2(\g)$ across the classification of complex simple Lie algebras.
Low-rank coincidences ($\mathfrak{so}_3\cong\mathfrak{sl}_2\cong\mathfrak{sp}_2$, $\mathfrak{so}_5\cong\mathfrak{sp}_4$, $\mathfrak{so}_6\cong\mathfrak{sl}_4$) give consistent values.
The value $2$, required by Eq.~\eqref{eq:app-casimir}, never occurs.}
\begin{ruledtabular}
\begin{tabular}{ll}
$\g$ & eigenvalues on $\Sym^2(\g)$, in units of $c_\g$ \\
\hline
$\mathfrak{sl}_2$ & $0,\;3$ \\
$\mathfrak{sl}_n$ ($n\ge3$) & $0,\;1,\;2-\tfrac2n,\;2+\tfrac2n$ \\
$\mathfrak{so}_n$ ($n\ge5$) & $0,\;\tfrac{n}{n-2},\;2-\tfrac4{n-2},\;2+\tfrac2{n-2}$ \\
$\mathfrak{sp}_{2n}$ ($n\ge2$) & $0,\;\tfrac{n}{n+1},\;2-\tfrac1{n+1},\;2+\tfrac2{n+1}$ \\
$\mathfrak{g}_2$ & $0,\;\tfrac76,\;\tfrac52$ \\
$\mathfrak{f}_4$ & $0,\;\tfrac43,\;\tfrac{13}9,\;\tfrac{20}9$ \\
$\mathfrak{e}_6$ & $0,\;\tfrac43,\;\tfrac32,\;\tfrac{13}6$ \\
$\mathfrak{e}_7$ & $0,\;\tfrac43,\;\tfrac{14}9,\;\tfrac{19}9$ \\
$\mathfrak{e}_8$ & $0,\;\tfrac43,\;\tfrac85,\;\tfrac{31}{15}$ \\
\end{tabular}
\end{ruledtabular}
\end{table}

\section{Reality condition; proofs of Theorem~\ref{thm:signature} and Corollary~\ref{cor:radical}}\label{app:B}

This appendix states the reality condition in components, proves Theorem~\ref{thm:signature} and the complete list below twelve dimensions, and proves Corollary~\ref{cor:radical}.

\subsection{Reality condition and extension}
An antilinear involutive automorphism $\sigma$~\cite{Knapp2002} specifies a real form $\gR=\mathrm{Fix}(\sigma)$; in components $\sigma(x)=S\bar x$ with $S\bar S=\mathbb{I}$.
A complex solution $h$ of Eq.~\eqref{eq:master} descends to a real symmetric form on $\gR$ if and only if it is $\sigma$-compatible, $\overline{h}=S^{T}hS$.
Choosing a basis $\{r^a\}$ of $\sigma$-fixed vectors, $r^a=T^a{}_b\,t^b$, the real form is $h^{\mathbb{R}}_{ab}=h(r^a,r^b)=(ThT^{T})_{ab}$, real and symmetric, and the structure constants in this basis are real.
Conversely, every real symmetric solution of Eq.~\eqref{eq:master} on $\gR$ extends $\mathbb{C}$-bilinearly to a $\sigma$-compatible solution on $\gC$; the real and complex problems are therefore equivalent.
For a swapped pair we choose Chevalley (real-structure-constant) bases so that $\sigma(x,y)=(\bar y,\bar x)$, i.e., $S=\left(\begin{smallmatrix}0&\mathbb{I}\\\mathbb{I}&0\end{smallmatrix}\right)$, with fixed locus $\{(x,\bar x)\}$.

\subsection{Proof of Theorem~\ref{thm:signature}}

\emph{(a) Witt bound.}
By Corollary~\ref{cor:semisimple}, $h^{(ii)}=0$ on a fixed ideal, so $h^{\mathbb{R}}$ vanishes identically on $\g_{i,\mathbb{R}}$---a totally isotropic subspace---of dimension $d_i=\dim_{\mathbb{C}}\g_i\ge3$, irrespective of all off-diagonal blocks.
The Witt index of a nondegenerate form of signature $(p,q)$ is $\min(p,q)$, so $\min(p,q)\ge d_i$. \hfill$\square$

\emph{(b) Single pair.}
With $m=0$, $n=1$, the only block allowed by Corollary~\ref{cor:semisimple} is the coupling $h=\left(\begin{smallmatrix}0&B\\B^{T}&0\end{smallmatrix}\right)$, and compatibility gives $\bar B=B^{T}$: $B$ is Hermitian, nondegenerate if and only if $h$ is.
On the fixed locus the form is $h^{\mathbb{R}}\big((x,\bar x),(x,\bar x)\big)=2x^{T}B\bar x=2x^{\dagger}\bar Bx$, and $\bar B=B^{T}$ is Hermitian with the same eigenvalues as $B$.
Diagonalizing $\bar B=U\mathrm{diag}(\lambda_i)U^{\dagger}$ and writing $z=U^{\dagger}x$ gives $2\sum_i\lambda_i[(\mathrm{Re}\,z_i)^2+(\mathrm{Im}\,z_i)^2]$: each eigenvalue contributes two real squares of its sign, so $(p,q)=(2\tilde p,2\tilde q)$; conversely $B=\mathrm{diag}(\pm1)$ realizes every even--even signature. \hfill$\square$

\emph{(c) Coupled pairs.}
Label the copies of pairs $\alpha\neq\beta$ by $(1,2)$ and $(3,4)$.
Corollary~\ref{cor:semisimple} leaves the cross blocks $K^{13},K^{14}$ arbitrary, and compatibility fixes their partners, $K^{24}=\overline{K^{13}}$, $K^{23}=\overline{K^{14}}$.
The real pairing of the fixed loci is
\begin{equation}\label{eq:app-crossN}
  F(x,u)=2\,\mathrm{Re}\big(x^{T}K^{13}u\big)+2\,\mathrm{Re}\big(x^{T}K^{14}\bar u\big):
\end{equation}
since every real-bilinear form on $\mathbb{C}^{e_\alpha}\times\mathbb{C}^{e_\beta}$ (as real spaces) splits uniquely into parts complex-linear and antilinear in $u$, Eq.~\eqref{eq:app-crossN} exhausts \emph{all} real $2e_\alpha\times2e_\beta$ coupling matrices.
Now fix a signature $(p,d-p)$; by $h\to-h$ we may take $p\le d/2$.
If $p$ is even, decouple the pairs and choose the $B_\alpha$ eigenvalues with $p/2$ entries $+1$: done by (b).
If $p$ is odd, choose the decoupled solution with $(p-1)/2$ eigenvalues $+1$, placing at least one $-1$ eigenvalue in each pair (this is possible, since the number of $-1$ eigenvalues is $(d-p+1)/2\ge2$ for $p\le d/2$ and $d\ge12$).
Then couple one real direction of a negative $2$-plane of the first pair to one of the second with strength $3$: that $2\times2$ block becomes $\left(\begin{smallmatrix}-2&3\\3&-2\end{smallmatrix}\right)$ with eigenvalues $\{1,-5\}$, converting exactly one negative direction into a positive one and leaving all others unchanged.
The result is nondegenerate of signature $(p,d-p)$. \hfill$\square$

\subsection{The complete list below twelve dimensions}
The building blocks of a real semisimple algebra of dimension $d\le11$ are fixed ideals of real dimension $d_i\in\{3,8,10\}$ (the complex simple dimensions not exceeding $11$) and swapped pairs of real dimension $6$ (the next pair dimension is $16$).
The realizable dimensions and decompositions are $d=3\;[3]$; $6\;[3{+}3$ or one pair$]$; $8\;[8]$; $9\;[3{+}3{+}3$ or $3{+}$pair$]$; $10\;[10]$; $11\;[3{+}8]$.
If the algebra is a single fixed ideal ($d=3,8,10$), Theorem~\ref{thm:nogo} forces $h=0$.
For $d=11=3{+}8$ the only nonvanishing blocks form a coupling $h=\left(\begin{smallmatrix}0&B\\B^{T}&0\end{smallmatrix}\right)$ with $B$ a real $3\times8$ matrix, so $\mathrm{rank}\,h=2\,\mathrm{rank}\,B\le6<11$: every solution is degenerate.
At $d=6$ and $d=9$ nondegenerate solutions exist: the single pair with invertible Hermitian $B$, which is case (b); two fixed three-dimensional ideals with $h=\left(\begin{smallmatrix}0&B\\B^{T}&0\end{smallmatrix}\right)$, $B\in GL(3,\mathbb{R})$, of signature $(3,3)$; and, at $d=9$, $\mathfrak{su}(2)^{\oplus3}$ with all couplings $\mathbb{I}_3$, i.e., $h=(J-\mathbb{I})\otimes\mathbb{I}_3$ with $J$ the all-ones $3\times3$ matrix, of eigenvalues $\{2,-1,-1\}$ (each threefold) and signature $(3,6)$.
Hence nondegenerate semisimple solutions with $d\le11$ occur exactly at $d=6$ and $d=9$, and by Theorem~\ref{thm:signature} none is Lorentzian. \hfill$\square$

\subsection{Restriction of Eq.~\eqref{eq:master} to the Levi subalgebra}
Write $\g=\gss\oplus\rad$ for the Levi decomposition as a vector space, with $\gss$ a subalgebra and $\rad$ an ideal; the brackets $[\gss,\rad]\subseteq\rad$ need not vanish.
Let $h$ solve Eq.~\eqref{eq:master} on $\g$ and take $x\in\gss$.
The contribution of the $\gss$ block of $h$ is $\sum_{ij}h^{ij}[s_i,[s_j,x]]$, with $\{s_i\}$ a basis of $\gss$; it lies in $\gss$ because $\gss$ is closed under the bracket, $[\gss,\gss]\subseteq\gss$.
The contributions of the coupling and $\rad$ blocks contain a generator of $\rad$ inside the brackets and lie in $\rad$, because $\rad$ is an ideal, $[\rad,\g]\subseteq\rad$.
Since $\gss$ and $\rad$ are linearly independent, the two parts vanish separately: the restriction $h|_{\gss}$ solves Eq.~\eqref{eq:master} on the semisimple algebra $\gss$ alone.
This step uses neither nondegeneracy nor the signature of $h$.
With $\{r_m\}$ a basis of $\rad$ and blocks $h^{ij}$, $h^{im}$, $h^{mn}$, the remaining parts are, for $t^c=x\in\gss$, the $\rad$ part
\begin{equation}\label{eq:app-radpart}
\begin{aligned}
&\sum_{i,m}h^{im}\big([s_i,[r_m,x]]+[r_m,[s_i,x]]\big)\\
&\quad+\sum_{m,n}h^{mn}[r_m,[r_n,x]]=0,
\end{aligned}
\end{equation}
and, for $t^c=y\in\rad$, the full equation
\begin{equation}\label{eq:app-radeq}
\begin{aligned}
&\sum_{i,j}h^{ij}[s_i,[s_j,y]]+\sum_{i,m}h^{im}\big([s_i,[r_m,y]]+[r_m,[s_i,y]]\big)\\
&\quad+\sum_{m,n}h^{mn}[r_m,[r_n,y]]=0.
\end{aligned}
\end{equation}
These parts constrain the couplings and involve $h|_{\gss}$ through the brackets between $\gss$ and $\rad$; Sec.~\ref{sec:radical} and Appendix~\ref{app:G} solve them for the semidirect algebras, and the proof of Corollary~\ref{cor:radical} below uses only the $\gss$ part.

\subsection{Proof of Corollary~\ref{cor:radical}}
As shown above, $h|_{\gss}$ solves Eq.~\eqref{eq:master} on $\gss$; within $\gss$, distinct simple ideals commute, so for $c\in\g_i$ the equation involves only the $\g_i$ block.
(i) If $\g_i$ is $\sigma$-fixed, Theorem~\ref{thm:nogo} gives $h^{(ii)}=0$, and its real points form a totally isotropic subspace of dimension $\ge3$, exceeding the Witt index $\min(1,d-1)=1$ of a Lorentzian form: contradiction.
So $\gss$ consists of swapped pairs; each has real dimension $2e_\alpha\ge6$, so two pairs would need $d\ge12$, and a single pair has $2e\le11$, forcing $e=3$ (the complex simple dimensions being $3,8,10,14,\dots$): $\gss=0$ or $\gss\cong\mathfrak{sl}(2,\mathbb{C})_\mathbb{R}\cong\mathfrak{so}(1,3)$.
(ii) Since $h|_{\gss}$ solves Eq.~\eqref{eq:master} on $\gss$, its $\mathbb{C}$-bilinear extension to $(\gss)_\mathbb{C}$ is a $\sigma$-compatible solution there, possibly degenerate.
For the single pair, the analysis of case (b) does not use nondegeneracy and applies here: the diagonal blocks vanish and the coupling $B$ is Hermitian, so the form on $\gss$ is $2x^\dagger\bar Bx$, whose positive index is $2\cdot\#\{\lambda_i(B)>0\}$ and whose kernel is $\{x:\bar Bx=0\}$, of real dimension $2\dim_{\mathbb{C}}\ker B$---both even numbers.
The positive index of the restriction of $h$ to any subspace is at most that of $h$, which is $1$; even and $\le1$ forces $0$, so $h|_{\gss}\preceq0$.
The kernel of $h|_{\gss}$ is a subspace of $\g$ on which $h$ vanishes identically, hence of dimension at most the Witt index $1$ of $h$; even and $\le1$ again forces $0$, so $h|_{\gss}$ is nondegenerate.
Negative semidefinite and nondegenerate means negative definite: $h(v,v)<0$ for every nonzero $v\in\gss$.
A vector with $h(v,v)\ge0$, timelike or null, therefore has a nonzero component in $\rad$. \hfill$\square$

\section{Pullback bounds for $d<k$ and the ten-dimensional model}\label{app:C}

This appendix treats configurations with $d<k$ and proves the ten-dimensional statements of Sec.~\ref{sec:lorentz}, for either target signature.
If the $k$ matrices span an algebra of dimension $d<k$, then $C$ has rank $d$, $h=C^{T}\eta C$ is the pullback of $\eta$, possibly degenerate, and Eq.~\eqref{eq:EOM} still implies Eq.~\eqref{eq:master} because $C$ is injective.
The same injectivity maps subspaces on which $h$ is positive (negative) definite onto subspaces on which $\eta$ is positive (negative) definite, and $h$-isotropic subspaces onto $\eta$-isotropic subspaces; hence $p(h)\le p(\eta)$ and every $h$-isotropic subspace has dimension at most the Witt index $\min(p_\eta,q_\eta)$.
Let $\gss$ be the semisimple part of $\gR$; by Appendix~\ref{app:B}, $h|_{\gss}$ solves Eq.~\eqref{eq:master} on $\gss$ whether or not $h$ is degenerate, so the block structure of Corollary~\ref{cor:semisimple} applies to it.
First, every $\sigma$-fixed simple ideal is $h$-isotropic (Theorem~\ref{thm:nogo}) of dimension $d_i\ge3$, while the Witt index of $\eta$ is $0$ (Euclidean) or $1$ (Lorentzian): no fixed ideal can occur.
Second, two swapped pairs require $d\ge12>10$, and a single pair on any simple algebra beyond $\mathfrak{sl}_2$ has real dimension $\ge16$; so $\gss$ is $0$ or the single $\mathfrak{sl}_2$ pair, $\mathfrak{sl}(2,\mathbb{C})_{\mathbb{R}}\cong\mathfrak{so}(1,3)$.
Third, for that pair $h|_{\gss}$ has the form $2x^\dagger\bar Bx$ of Theorem~\ref{thm:signature}(b), with positive index $2\tilde p$ and kernel of even dimension $2\dim_{\mathbb{C}}\ker B$.
The kernel is $h$-isotropic, hence of dimension $\le1$, hence $0$: $B$ is nondegenerate.
The positive index satisfies $2\tilde p\le p(\eta)\le1$, so $\tilde p=0$.
Hence $h$ is negative definite on $\gss$, of signature $(0,6)$, and the semisimple part occupies only spacelike directions of the target space.
In particular, a purely semisimple configuration of the ten-dimensional model of either signature is the $\mathfrak{so}(1,3)$ saddle of Ref.~\cite{Liao:2025yfb}, and no semisimple configuration reaches the timelike direction.
This proves the ten-dimensional statements of Sec.~\ref{sec:lorentz}.
The same argument shows that no algebra of dimension $\le5$ with nonzero semisimple part---necessarily a fixed $\mathfrak{su}(2)$ or $\mathfrak{sl}(2,\mathbb{R})$---hosts a configuration of either model, which explains the outcome of the survey of Ref.~\cite{Chatzistavrakidis:2011su}. \hfill$\square$

\section{Explicit Lorentzian solutions}\label{app:D}

This appendix records the explicit solutions used in Secs.~\ref{sec:lorentz} and~\ref{sec:radical}; we have checked each of them by substituting into Eq.~\eqref{eq:master}.
Section~\ref{sec:radical} displays the semidirect solutions and Appendix~\ref{app:G} treats them.

\subsection{Purely semisimple, $d=12$: sharpness of Theorem~\ref{thm:main}}
Take $\gC=\mathfrak{sl}_2^{(1)}\oplus\cdots\oplus\mathfrak{sl}_2^{(4)}$ in Chevalley bases $\{H,E,F\}$, with $\sigma$ swapping $1{\leftrightarrow}2$ and $3{\leftrightarrow}4$, so $\gR=\mathfrak{sl}(2,\mathbb{C})_{\mathbb{R}}^{\oplus2}\cong\mathfrak{so}(1,3)^{\oplus2}$.
Set the intra-pair blocks $h^{(12)}=h^{(34)}=-\mathbb{I}_3$ (Hermitian), the cross-pair blocks $h^{(13)}_{HH}=h^{(14)}_{HH}=\tfrac34$ with $\sigma$-partners $h^{(24)}=\overline{h^{(13)}}$, $h^{(23)}=\overline{h^{(14)}}$, and all other entries---in particular all four diagonal blocks---to zero, so Corollary~\ref{cor:semisimple} and $\sigma$-compatibility hold by construction.
On the fixed locus, parametrized by $(x,\bar x,u,\bar u)$, the quadratic form is
\begin{equation}\label{eq:sm-d12}
  Q=-2|x|^{2}-2|u|^{2}+6\,\mathrm{Re}(x_{H})\,\mathrm{Re}(u_{H}),
\end{equation}
diagonal except for one $2\times2$ block $\left(\begin{smallmatrix}-2&3\\3&-2\end{smallmatrix}\right)$ with eigenvalues $\{1,-5\}$: the signature is $(1,11)$, nondegenerate and Lorentzian, with the timelike direction along the sum of the two $H$-directions---a boost generator from each $\mathfrak{so}(1,3)$ factor, tilted timelike by the cross-coupling.
In the basis of rotations and boosts, an equivalent solution is $h=-\mathbb{I}_{12}$ plus a coupling $h_{K_3^{(1)}K_3^{(2)}}=c$ between the two $K_3$ boosts, which is Lorentzian for $|c|>1$.

\subsection{Semisimple $\oplus$ solvable, $d=8$}
Take $\gR=\mathfrak{so}(1,3)\oplus\mathfrak{aff}(1)$, $\mathfrak{aff}(1)=\{X,Y:[X,Y]=Y\}$.
On $\mathfrak{aff}(1)$, Eq.~\eqref{eq:master} forces $h_{XX}=h_{XY}=0$; set $h_{YY}=-1$.
Take the pair block $h^{(12)}=-\mathbb{I}_3$ as above, and the single cross-coupling $h_{H^{(1)}X}=h_{H^{(2)}X}=\tfrac34$ (its own $\sigma$-partner).
The quadratic form on the fixed locus is
\begin{equation}\label{eq:sm-d8}
  Q=-2|x|^{2}-Y^{2}+3\,\mathrm{Re}(x_{H})\,X,
\end{equation}
whose $(\mathrm{Re}\,x_{H},X)$ block $\left(\begin{smallmatrix}-2&3/2\\3/2&0\end{smallmatrix}\right)$ has negative determinant: signature $(1,7)$, nondegenerate and Lorentzian, with the timelike direction supported on the null radical generator $X$.

\subsection{Full rank, $d=10$: a Lorentzian background of the physical model}
Append to the $d=8$ solution two commuting generators $Z_{1,2}$ with $h_{Z_iZ_j}=-\delta_{ij}$ and no other couplings; abelian directions solve Eq.~\eqref{eq:master} identically, so
\begin{equation}\label{eq:sm-d10}
  \gR=\mathfrak{so}(1,3)\oplus\mathfrak{aff}(1)\oplus\mathbb{R}^{2},\qquad
  Q_{10}=Q-Z_1^2-Z_2^2 ,
\end{equation}
with $Q$ as in Eq.~\eqref{eq:sm-d8}, is a nondegenerate solution of signature $(1,9)$ on an algebra of dimension exactly ten.
Choosing any $C\in GL(10,\mathbb{R})$ with $C^{T}\eta C=h$ (which exists by Sylvester's law of inertia, since $h$ and $\eta$ have the same signature) and any faithful unitary representation $\pi$ of $\gR$, the ten matrices $A^\mu=i\,C^\mu{}_a\,\pi(t^a)$ form a full-rank Lorentzian background of the ten-dimensional model: a spacelike $\mathfrak{so}(1,3)$ sector, time on the null solvable generator $X$, and two commutative directions.
On $\mathfrak{aff}(1)$ one may take $\pi(X)=\partial_s$ and $\pi(Y)=ie^{s}$ on $L^2(\mathbb{R})$, so that the timelike matrix is a momentum-like operator with continuous spectrum; for the two commutative directions, a spread spectrum requires a direct integral over their characters (Sec.~\ref{sec:radical}).

\section{Infrared regulators are fixed to the adjoint Casimir}\label{app:E}

This appendix derives the multiplier condition used in Secs.~\ref{sec:complex} and~\ref{sec:disc}.

The Lorentzian model of Refs.~\cite{Kim:2011ts,Kim:2012mw} carries infrared cutoffs on $\Tr(A^0)^2$ and $\Tr(A^i)^2$, implemented by Lagrange multipliers.
Any such quadratic constraint terms modify the equations of motion to the mass-deformed form
\begin{equation}\label{eq:sm-knt}
  \eta_{\nu\rho}\,[A^\nu,[A^\rho,A^\mu]]=\lambda_{(\mu)}A^\mu
  \quad(\text{no sum over }\mu),
\end{equation}
with one multiplier for each constrained group of directions.
Take any Lie-algebraic ansatz in which the occupied directions span a \emph{simple} algebra and $h=C^{T}\eta\,C$ is proportional to the inverse Killing form---as on the $\mathfrak{su}(2)$ and $\mathfrak{su}(1,1)$ branches of Ref.~\cite{Kim:2011ts}.
Then, exactly as in Appendix~\ref{app:A}, the left-hand side of Eq.~\eqref{eq:sm-knt} is $C^\mu{}_c\,h_{ab}[t^a,[t^b,t^c]]=c_h\,A^\mu$ with $c_h\propto c_\g\neq0$ the adjoint Casimir eigenvalue, acting identically on every occupied direction.
Equation~\eqref{eq:sm-knt} therefore holds if and only if every multiplier equals this one nonzero constant, $\lambda_{(\mu)}=c_h$ for all occupied $\mu$: the regulator terms are mass terms fixed to the adjoint Casimir.
They cannot be switched off---in the limit $\lambda\to0$ the branch disappears, in agreement with Theorem~\ref{thm:nogo}---and, conversely, a genuine mass deformation converts the Casimir obstruction into a solvability condition, which is why semisimple Lie-algebra solutions populate the mass-deformed model~\cite{GoharaSako2025,Berenstein:2002jq}.

The expanding branch of Ref.~\cite{Kim:2011ts} is the solvable algebra $[A^0,A^i]=i\sqrt{\lambda}\,A^i$, $[A^i,A^j]=0$, with $h$ the flat metric.
For it, $-[A^0,[A^0,A^i]]=\lambda A^i$ and $[A^j,[A^j,A^0]]=0$, so Eq.~\eqref{eq:sm-knt} holds with a nonzero spatial multiplier, $\lambda_{(i)}\propto\lambda$, and a vanishing temporal one, $\lambda_{(0)}=0$.
In the undeformed equation, $\lambda=0$, the same algebra requires $h_{00}=h_{0i}=0$ (Sec.~\ref{sec:radical}): the time generator must be null, and with the flat metric the branch again disappears.
The regulator thus plays a different role on the two branches: it cancels the Casimir on the simple ones and it replaces the null coupling on the solvable one.

\section{Lie-algebraic solutions with $d\le11$}\label{app:F}

This appendix summarizes the known Lie-algebraic solutions with $d\le11$ (Table~\ref{tab:landscape}).
The semisimple part follows from Sec.~\ref{sec:real} and is complete.
For the solvable part, Ref.~\cite{Chatzistavrakidis:2011su} scanned all Lie algebras of dimension up to five and all nilpotent Lie algebras of dimension six with the generators placed along coordinate axes, that is, with $h$ equal to the flat metric in the tabulated basis; every solution found there is nilpotent or solvable, and Appendix~\ref{app:C} explains why no algebra with a nonzero semisimple part can appear in that range.
That scan is complete for the flat metric but not for general $h$: the solvable algebras $\mathfrak{aff}(1)\oplus\mathbb{R}$ (Sec.~\ref{sec:radical}) and $\mathfrak{aff}(1)\oplus\mathbb{R}^3$ admit nondegenerate Lorentzian solutions, of signatures $(1,2)$ and $(1,4)$, in which the null generator $X$ is coupled to a commuting direction, whereas with the flat metric $\mathfrak{aff}(1)$ fails.
Abelian and $2$-step nilpotent algebras solve Eq.~\eqref{eq:master} identically in every dimension, and nilpotent algebras also underlie the nilmanifold backgrounds of Ref.~\cite{ChatzistavrakidisJonke2012}.
Solvable algebras of dimension six and higher are not classified; for each fixed, finite-dimensional Lie algebra, Eq.~\eqref{eq:master} is a finite linear system for $h$, so one can extend the list case by case, as we did for $\mathfrak{aff}(1)$, $\mathfrak{e}(2)$, and $\mathfrak{e}(1,1)$ here.
The semidirect rows follow from Sec.~\ref{sec:radical} and Appendix~\ref{app:G}.

\begin{table*}[tb]
\caption{\label{tab:landscape}Lie-algebraic solutions of Eq.~\eqref{eq:master} with $d\le11$, organized by the Levi decomposition.
``Euclidean'' indicates whether a nondegenerate definite solution exists on the algebra; ``Lorentzian'' whether a nondegenerate solution of signature $(1,d-1)$ exists.
The semisimple rows are complete (Sec.~\ref{sec:real}); the survey of Ref.~\cite{Chatzistavrakidis:2011su} covers the solvable sector for $d\le5$, and nilpotent algebras at $d=6$, with $h$ fixed to the flat metric, and the solvable sector for general $h$ is open beyond the rows listed.}
\begin{ruledtabular}
\begin{tabular}{llllcc}
Algebra & $d$ & Solutions of Eq.~\eqref{eq:master} & Signatures & Euclidean & Lorentzian \\
\hline
abelian $\mathbb{R}^d$ & any & $h$ unconstrained & all $(p,q)$ & yes & yes \\
$2$-step nilpotent ($\mathfrak{h}_3$, etc.) & any & $h$ unconstrained & all $(p,q)$ & yes & yes \\
$\mathfrak{aff}(1)$ & $2$ & $h_{XX}=h_{XY}=0$ & degenerate only & no & no \\
$\mathfrak{aff}(1)\oplus\mathbb{R}^{n}$, $n\ge1$ & $2{+}n$ & $h_{XX}=h_{XY}=0$, rest free & $(1,1{+}n)$, etc. & no & yes \\
$\mathfrak{e}(2)$, $\mathfrak{e}(1,1)$ & $3$ & $h_{JJ}=h_{Ja}=0$, $P$ block free & degenerate only\footnotemark[1] & no & no \\
any algebra, $d\le5$; nilpotent, $d=6$ & $\le6$ & flat $h$: Ref.~\cite{Chatzistavrakidis:2011su} & nilpotent or solvable & \multicolumn{2}{c}{general $h$: open} \\
solvable, non-nilpotent & $6$--$11$ & not classified & open & open & open \\
single fixed simple ideal & $3,8,10$ & $h=0$ & --- & no & no \\
$\mathfrak{so}(1,3)$ (one pair) & $6$ & Hermitian $B$ & $(2\tilde p,2\tilde q)$ & yes: $(0,6)$ & no \\
two $3$-dim fixed ideals & $6$ & invertible coupling & $(3,3)$ & no & no \\
three $3$-dim ideals, or $3$-dim${}\oplus{}$pair & $9$ & couplings & $\min(p,q)\ge3$; e.g.\ $(3,6)$ & no & no \\
$\mathfrak{so}(1,3)\oplus\mathfrak{h}_3$ & $9$ & block sums & $(0,9)$, $(1,8)$, etc. & yes & yes \\
$\mathfrak{so}(1,3)\oplus\mathfrak{aff}(1)$ & $8$ & null coupling & $(1,7)$ & no & yes \\
$\mathfrak{so}(1,3)\oplus\mathfrak{e}(2)$ & $9$ & null coupling & $(1,8)$ & no & yes \\
$\mathfrak{so}(1,3)\oplus\mathfrak{aff}(1)\oplus\mathbb{R}^{2}$ & $10$ & Eq.~\eqref{eq:sm-d10} & $(1,9)$ & no & yes \\
$\mathfrak{so}(1,3)\ltimes\mathbb{R}^{4}$ (vector: Poincar\'e) & $10$ & $h$ on translations only & degenerate only & no & no \\
$\mathfrak{so}(1,3)\ltimes\mathbb{R}^{4}$ (spinor) & $10$ & Eq.~\eqref{eq:weylsol} & $(1,9)$, $(0,10)$, etc. & yes & yes \\
$\mathfrak{so}(1,3)\ltimes\mathfrak{h}_5$ (spinor, central $Z$) & $11$ & Eq.~\eqref{eq:heissol} & $(1,10)$, $(0,11)$, etc. & yes & yes \\
\end{tabular}
\end{ruledtabular}
\footnotetext[1]{With $J$ placed along a null direction of the target space, the degenerate solution is the propagating fuzzy cylinder of Ref.~\cite{Steinacker:2011wb}.}
\end{table*}

\section{Semidirect sums: proof of Theorem~\ref{thm:semi}}\label{app:G}

This appendix proves Theorem~\ref{thm:semi}.

\subsection{Representation bound}
The complexification $\mathfrak{so}(1,3)_{\mathbb{C}}\cong\mathfrak{sl}(2,\mathbb{C})\oplus\mathfrak{sl}(2,\mathbb{C})$ is the swapped pair of Sec.~\ref{sec:real}, with irreducible complex representations $(j_1,j_2)$ of dimension $(2j_1+1)(2j_2+1)$, and complex conjugation exchanges the two ideals, so $(j_1,j_2)^{*}\cong(j_2,j_1)$.
A real representation complexifies to one isomorphic to its own conjugate, hence to a sum of $(j,j)$'s and of pairs $(j_1,j_2)\oplus(j_2,j_1)$.
A one-dimensional representation is trivial, because $\mathfrak{so}(1,3)$ equals its own commutator and the target is abelian.
The smallest nontrivial options are $(\tfrac12,\tfrac12)$, the vector representation $V$, and $(\tfrac12,0)\oplus(0,\tfrac12)$, the Weyl spinor $W$ regarded as real; both have real dimension four, and the next, $(1,0)\oplus(0,1)$, has dimension six.
Every nontrivial real representation therefore has dimension at least four.
For $d\le9$, $\dim\rad\le3$ and the brackets $[\gss,\rad]$ vanish, which is statement (i).

\subsection{Abelian radical at $d=10$}
For $x\in\gss$ and $u,v\in\rad$, the Jacobi identity $[x,[u,v]]=[[x,u],v]+[u,[x,v]]$ says that the bracket of $\rad$ is an equivariant map $\Lambda^{2}\rad\to\rad$, from the antisymmetric pairs $u\wedge v$ to the radical.
At $d=10$ the radical is $V$ or $W$, and $\Lambda^{2}V_{\mathbb{C}}=(1,0)\oplus(0,1)$ while $\Lambda^{2}W_{\mathbb{C}}=(0,0)\oplus(\tfrac12,\tfrac12)\oplus(0,0)$: neither shares an irreducible constituent with the representation itself, so by Schur's lemma the bracket vanishes and the radical is abelian.

\subsection{The equations for $c\in\rad$ and statement (ii)}
With abelian $\rad$ and vanishing couplings, the only condition beyond the standalone $\mathfrak{so}(1,3)$ equation is $\sum_{ij}h^{ij}\rho(s_i)\rho(s_j)=0$ on $\rad$, with $h^{ij}$ the $\gss$ block, $\{s_i\}$ a basis of $\gss$, and $\rho$ the representation.
By Corollary~\ref{cor:semisimple}, every solution $h^{ij}$ couples only the two ideals of the pair, so this operator is a sum of mixed products of the two chiral factors.
On $W$ each irreducible half is annihilated by one factor and the operator vanishes for every solution; concretely, Eq.~\eqref{eq:weylact} gives $\rho(K_i)=-i\rho(J_i)$, so $\sum_i\rho(K_i)^{2}=-\sum_i\rho(J_i)^{2}$.
On $V$ both factors act and the operator does not vanish: for $h^{ij}=\delta$ it equals $\mathrm{diag}(3,-1,-1,-1)$ on the basis $(P_0,P_i)$, and in general it is proportional to $\sum_{\alpha\beta}B_{\alpha\beta}\,\sigma_\alpha\otimes\sigma_\beta$ on $V_{\mathbb{C}}\cong\mathbb{C}^2\otimes\mathbb{C}^2$, which vanishes only for $B=0$ because the products $\sigma_\alpha\otimes\sigma_\beta$ are linearly independent.
A direct computation of the solution space confirms this and excludes the entire nine-dimensional family, leaving only solutions supported on the translations, of rank at most four; the Poincar\'e algebra therefore admits no nondegenerate solution.
Solving the conditions on the couplings as well, the full solution space on the spinor algebra has dimension $31$, decomposing as $9+12+10$ into the $\mathfrak{so}(1,3)$ family, the couplings, and the free spinor block, and Eq.~\eqref{eq:weylsol} is a nondegenerate Lorentzian member.

\subsection{Nonabelian radicals at $d=11$ and statement (iii)}
A five-dimensional representation space on which $\mathfrak{so}(1,3)$ acts nontrivially is $M\oplus\mathbb{R}Z$ with $M\in\{V,W\}$ and $Z$ trivial, since no nontrivial real representation of dimension five exists.
The equivariant brackets are: $\Lambda^{2}M\to M$, zero (above); $\Lambda^{2}M\to\mathbb{R}Z$, zero on $V$ and the two-dimensional space of invariant forms $\omega\in\mathrm{span}\{\mathrm{Re}\,\varepsilon,\mathrm{Im}\,\varepsilon\}$ on $W$; and $[Z,u]=\varphi u$ with $\varphi$ an equivariant endomorphism.
Solvability of the radical excludes $\omega\neq0$ together with $\varphi\neq0$ (an equivariant $\varphi\neq0$ is invertible, and $[\rad,\rad]$ would then be all of $\rad$), and all nonzero $\omega$ are equivalent, since multiplication by a unit complex number on the doublet rotates $\varepsilon$ by a phase.
Direct computation of the full solution spaces gives the following.
The central extension $[Q_a,Q_b]=\omega_{ab}Z$ has a $36$-dimensional solution space, decomposing as $9+12+4+11$ into the $\mathfrak{so}(1,3)$ family, the couplings, the allowed $Q$ blocks, and the free $Z$ row, and admits nondegenerate solutions; the scaling on $V$ has a $10$-dimensional space with every solution of rank at most $4$; the scalings on $W$, evaluated at $\varphi=1$, $\varphi=i$, and a generic point of the family, have $31$-dimensional spaces with every solution of rank at most $10$.
For the central extension, the $Z$-component of the equations for $c=x\in\gss$ reads $\mathrm{tr}(\omega\,\rho(x)\,h_{QQ})=0$ for all $x$; these are six conditions, with a four-dimensional space of allowed $Q$ blocks.
The blocks $h_{QQ}=\pm\mathbb{I}_4$ satisfy them, because the form $\varepsilon(v,\rho(x)v)$ is complex-bilinear, so its value on $iv$ is minus its value on $v$ and the trace over the real basis cancels in pairs.
Since $Z$ is central, no entry of $h$ in its row enters any equation, and Eq.~\eqref{eq:heissol} follows, of signature $(1,10)$; the tilted solution with $h_{Q_1Z}\neq0$ of Sec.~\ref{sec:radical} follows in the same way.
This is statement (iii).

\bibliography{refs}

@article{Ishibashi:1996xs,
  author        = "Ishibashi, N. and Kawai, H. and Kitazawa, Y. and Tsuchiya, A.",
  title         = "{A large-$N$ reduced model as superstring}",
  journal       = "Nucl. Phys. B",
  volume        = "498",
  pages         = "467--491",
  year          = "1997",
  doi           = "10.1016/S0550-3213(97)00290-3",
  eprint        = "hep-th/9612115",
  archivePrefix = "arXiv"
}

@article{Aoki:1998vn,
  author        = "Aoki, H. and Iso, S. and Kawai, H. and Kitazawa, Y. and Tada, T.",
  title         = "{Space-time structures from IIB matrix model}",
  journal       = "Prog. Theor. Phys.",
  volume        = "99",
  pages         = "713--746",
  year          = "1998",
  doi           = "10.1143/PTP.99.713",
  eprint        = "hep-th/9802085",
  archivePrefix = "arXiv"
}

@article{Kim:2011cr,
  author        = "Kim, S.-W. and Nishimura, J. and Tsuchiya, A.",
  title         = "{Expanding (3+1)-dimensional universe from a Lorentzian matrix model for superstring theory in (9+1)-dimensions}",
  journal       = "Phys. Rev. Lett.",
  volume        = "108",
  pages         = "011601",
  year          = "2012",
  doi           = "10.1103/PhysRevLett.108.011601",
  eprint        = "1108.1540",
  archivePrefix = "arXiv",
  primaryClass  = "hep-th"
}

@article{Anagnostopoulos:2022dak,
  author        = "Anagnostopoulos, K. N. and Azuma, T. and Hatakeyama, K. and Hirasawa, M. and Ito, Y. and Nishimura, J. and Papadoudis, S. K. and Tsuchiya, A.",
  title         = "{Progress in the numerical studies of the type IIB matrix model}",
  journal       = "Eur. Phys. J. ST",
  volume        = "232",
  pages         = "3681--3695",
  year          = "2023",
  doi           = "10.1140/epjs/s11734-023-00849-x",
  eprint        = "2210.17537",
  archivePrefix = "arXiv",
  primaryClass  = "hep-th"
}

@article{Kim:2011ts,
  author        = "Kim, S.-W. and Nishimura, J. and Tsuchiya, A.",
  title         = "{Expanding universe as a classical solution in the Lorentzian matrix model for nonperturbative superstring theory}",
  journal       = "Phys. Rev. D",
  volume        = "86",
  pages         = "027901",
  year          = "2012",
  doi           = "10.1103/PhysRevD.86.027901",
  eprint        = "1110.4803",
  archivePrefix = "arXiv",
  primaryClass  = "hep-th"
}

@article{Kim:2012mw,
  author        = "Kim, S.-W. and Nishimura, J. and Tsuchiya, A.",
  title         = "{Late time behaviors of the expanding universe in the IIB matrix model}",
  journal       = "JHEP",
  volume        = "10",
  number        = "2012",
  pages         = "147",
  year          = "2012",
  doi           = "10.1007/JHEP10(2012)147",
  eprint        = "1208.0711",
  archivePrefix = "arXiv",
  primaryClass  = "hep-th"
}

@article{Steinacker:2017vqw,
  author        = "Steinacker, H. C.",
  title         = "{Cosmological space-times with resolved Big Bang in Yang-Mills matrix models}",
  journal       = "JHEP",
  volume        = "02",
  number        = "2018",
  pages         = "033",
  year          = "2018",
  doi           = "10.1007/JHEP02(2018)033",
  eprint        = "1709.10480",
  archivePrefix = "arXiv",
  primaryClass  = "hep-th"
}

@article{Chatzistavrakidis:2011su,
  author        = "Chatzistavrakidis, A.",
  title         = "{On Lie-algebraic solutions of the type IIB matrix model}",
  journal       = "Phys. Rev. D",
  volume        = "84",
  pages         = "106010",
  year          = "2011",
  doi           = "10.1103/PhysRevD.84.106010",
  eprint        = "1108.1107",
  archivePrefix = "arXiv",
  primaryClass  = "hep-th"
}

@article{GoharaSako2025,
  author        = "Gohara, J. and Sako, A.",
  title         = "{Quantization of Lie-Poisson algebra and Lie algebra solutions of mass-deformed type IIB matrix model}",
  journal       = "J. Math. Phys.",
  volume        = "67",
  pages         = "022301",
  year          = "2026",
  eprint        = "2503.24060",
  archivePrefix = "arXiv",
  primaryClass  = "hep-th"
}

@article{HanadaKawaiKimura2006,
  author        = "Hanada, M. and Kawai, H. and Kimura, Y.",
  title         = "{Describing curved spaces by matrices}",
  journal       = "Prog. Theor. Phys.",
  volume        = "114",
  pages         = "1295--1316",
  year          = "2006",
  doi           = "10.1143/PTP.114.1295",
  eprint        = "hep-th/0508211",
  archivePrefix = "arXiv"
}

@article{Tsuchiya2024reg,
  author        = "Hattori, K. and Mizuno, Y. and Tsuchiya, A.",
  title         = "{Regularization of matrices in the covariant derivative interpretation of matrix models}",
  journal       = "PTEP",
  volume        = "2024",
  pages         = "123B06",
  year          = "2024",
  doi           = "10.1093/ptep/ptae180",
  eprint        = "2410.13414",
  archivePrefix = "arXiv",
  primaryClass  = "hep-th"
}

@article{HoKawaiSteinacker2026,
  author        = "Ho, P.-M. and Kawai, H. and Steinacker, H. C.",
  title         = "{General relativity in IIB matrix model}",
  journal       = "JHEP",
  volume        = "02",
  number        = "2026",
  pages         = "070",
  year          = "2026",
  eprint        = "2509.06646",
  archivePrefix = "arXiv",
  primaryClass  = "hep-th"
}

@article{SperlingSteinacker2019,
  author        = "Sperling, M. and Steinacker, H. C.",
  title         = "{Covariant cosmological quantum space-time, higher-spin and gravity in the IKKT matrix model}",
  journal       = "JHEP",
  volume        = "07",
  number        = "2019",
  pages         = "010",
  year          = "2019",
  doi           = "10.1007/JHEP07(2019)010",
  eprint        = "1901.03522",
  archivePrefix = "arXiv",
  primaryClass  = "hep-th"
}

@article{Myers:1999ps,
  author        = "Myers, R. C.",
  title         = "{Dielectric branes}",
  journal       = "JHEP",
  volume        = "12",
  number        = "1999",
  pages         = "022",
  year          = "1999",
  doi           = "10.1088/1126-6708/1999/12/022",
  eprint        = "hep-th/9910053",
  archivePrefix = "arXiv"
}

@article{Steinacker:2010rh,
  author        = "Steinacker, H.",
  title         = "{Emergent geometry and gravity from matrix models: an introduction}",
  journal       = "Class. Quant. Grav.",
  volume        = "27",
  pages         = "133001",
  year          = "2010",
  doi           = "10.1088/0264-9381/27/13/133001",
  eprint        = "1003.4134",
  archivePrefix = "arXiv",
  primaryClass  = "hep-th"
}

@article{Madore:1991bw,
  author        = "Madore, J.",
  title         = "{The fuzzy sphere}",
  journal       = "Class. Quant. Grav.",
  volume        = "9",
  pages         = "69--88",
  year          = "1992",
  doi           = "10.1088/0264-9381/9/1/008"
}

@article{Berenstein:2002jq,
  author        = "Berenstein, D. E. and Maldacena, J. M. and Nastase, H. S.",
  title         = "{Strings in flat space and pp waves from $\mathcal{N}=4$ super Yang-Mills}",
  journal       = "JHEP",
  volume        = "04",
  number        = "2002",
  pages         = "013",
  year          = "2002",
  doi           = "10.1088/1126-6708/2002/04/013",
  eprint        = "hep-th/0202021",
  archivePrefix = "arXiv"
}

@article{Hartnoll:2024csr,
  author        = "Hartnoll, S. A. and Liu, J.",
  title         = "{The polarised IKKT matrix model}",
  journal       = "JHEP",
  volume        = "03",
  number        = "2025",
  pages         = "060",
  year          = "2025",
  doi           = "10.1007/JHEP03(2025)060",
  eprint        = "2409.18706",
  archivePrefix = "arXiv",
  primaryClass  = "hep-th"
}

@article{Komatsu:2024bop,
  author        = "Komatsu, S. and Martina, A. and Penedones, J. and Vuignier, A. and Zhao, X.",
  title         = "{Einstein gravity from a matrix integral -- Part I}",
  journal       = "",
  eprint        = "2410.18173",
  archivePrefix = "arXiv",
  primaryClass  = "hep-th",
  year          = "2024"
}

@article{Banks:1996vh,
  author        = "Banks, T. and Fischler, W. and Shenker, S. H. and Susskind, L.",
  title         = "{M theory as a matrix model: A conjecture}",
  journal       = "Phys. Rev. D",
  volume        = "55",
  pages         = "5112--5128",
  year          = "1997",
  doi           = "10.1103/PhysRevD.55.5112",
  eprint        = "hep-th/9610043",
  archivePrefix = "arXiv"
}

@article{Taylor:1996ik,
  author        = "Taylor, W.",
  title         = "{D-brane field theory on compact spaces}",
  journal       = "Phys. Lett. B",
  volume        = "394",
  pages         = "283--287",
  year          = "1997",
  doi           = "10.1016/S0370-2693(97)00033-6",
  eprint        = "hep-th/9611042",
  archivePrefix = "arXiv"
}

@article{Douglas:2001ba,
  author        = "Douglas, M. R. and Nekrasov, N. A.",
  title         = "{Noncommutative field theory}",
  journal       = "Rev. Mod. Phys.",
  volume        = "73",
  pages         = "977--1029",
  year          = "2001",
  doi           = "10.1103/RevModPhys.73.977",
  eprint        = "hep-th/0106048",
  archivePrefix = "arXiv"
}

@article{Snyder:1946qz,
  author        = "Snyder, H. S.",
  title         = "{Quantized space-time}",
  journal       = "Phys. Rev.",
  volume        = "71",
  pages         = "38--41",
  year          = "1947",
  doi           = "10.1103/PhysRev.71.38"
}

@article{Deligne1996,
  author        = "Deligne, P.",
  title         = "{La s\'erie exceptionnelle de groupes de Lie}",
  journal       = "C. R. Acad. Sci. Paris S\'er. I Math.",
  volume        = "322",
  pages         = "321--326",
  year          = "1996"
}

@article{LandsbergManivel2006,
  author        = "Landsberg, J. M. and Manivel, L.",
  title         = "{A universal dimension formula for complex simple Lie algebras}",
  journal       = "Adv. Math.",
  volume        = "201",
  pages         = "379--407",
  year          = "2006",
  doi           = "10.1016/j.aim.2005.02.007",
  eprint        = "math/0401296",
  archivePrefix = "arXiv"
}

@article{Mkrtchyan2012,
  author        = "Mkrtchyan, R. L. and Sergeev, A. N. and Veselov, A. P.",
  title         = "{Casimir eigenvalues for universal Lie algebra}",
  journal       = "J. Math. Phys.",
  volume        = "53",
  pages         = "102106",
  year          = "2012",
  doi           = "10.1063/1.4757763",
  eprint        = "1105.0115",
  archivePrefix = "arXiv",
  primaryClass  = "math.RT"
}

@book{Knapp2002,
  author        = "Knapp, A. W.",
  title         = "{Lie Groups Beyond an Introduction}",
  edition       = "2nd",
  series        = "Progress in Mathematics",
  volume        = "140",
  publisher     = {Birkh{\"a}user},
  address       = "Boston",
  year          = "2002"
}

@article{Liao:2025yfb,
  author        = "Liao, H. and Maeta, R.",
  title         = "{New type of saddle in the Euclidean IKKT matrix model and its emergent geometry}",
  journal       = "Phys. Rev. D",
  volume        = "114",
  pages         = "026021",
  year          = "2026",
  doi           = "10.1103/gfcj-9vb3",
  eprint        = "2512.03161",
  archivePrefix = "arXiv",
  primaryClass  = "hep-th"
}

@article{Nishimura:2022alt,
  author        = "Nishimura, J.",
  title         = "{Signature change of the emergent space-time in the IKKT matrix model}",
  journal       = "PoS",
  volume        = "CORFU2021",
  pages         = "255",
  year          = "2022",
  doi           = "10.22323/1.406.0255",
  eprint        = "2205.04726",
  archivePrefix = "arXiv",
  primaryClass  = "hep-th"
}

@article{ChatzistavrakidisJonke2012,
  author        = {Chatzistavrakidis, Athanasios and Jonke, Larisa},
  title         = {Matrix theory compactifications on twisted tori},
  journal       = {Phys. Rev. D},
  volume        = {85},
  pages         = {106013},
  year          = {2012},
  doi           = {10.1103/PhysRevD.85.106013},
  eprint        = {1202.4310},
  archivePrefix = {arXiv},
  primaryClass  = {hep-th}
}

@article{Mackey1952,
  author  = {Mackey, George W.},
  title   = {Induced representations of locally compact groups. {I}},
  journal = {Ann. Math.},
  volume  = {55},
  pages   = {101--139},
  year    = {1952}
}

@book{Folland1989,
  author    = {Folland, Gerald B.},
  title     = {Harmonic Analysis in Phase Space},
  series    = {Annals of Mathematics Studies},
  volume    = {122},
  publisher = {Princeton University Press},
  address   = {Princeton},
  year      = {1989}
}

@article{Ito2014rg,
  author  = {Ito, Y. and Kim, S.-W. and Koizuka, Y. and Nishimura, J. and Tsuchiya, A.},
  title   = {{A renormalization group method for studying the early universe in the Lorentzian IIB matrix model}},
  journal = {PTEP},
  volume  = {2014},
  pages   = {083B01},
  year    = {2014},
  eprint  = {1312.5415},
  archivePrefix = {arXiv}
}

@article{NishimuraTsuchiya2019cl,
  author  = {Nishimura, J. and Tsuchiya, A.},
  title   = {{Complex Langevin analysis of the space-time structure in the Lorentzian type IIB matrix model}},
  journal = {JHEP},
  volume  = {06},
  pages   = {077},
  year    = {2019},
  eprint  = {1904.05919},
  archivePrefix = {arXiv}
}

@article{Aoki2019str,
  author  = {Aoki, T. and Hirasawa, M. and Ito, Y. and Nishimura, J. and Tsuchiya, A.},
  title   = {{On the structure of the emergent 3d expanding space in the Lorentzian type IIB matrix model}},
  journal = {PTEP},
  volume  = {2019},
  pages   = {093B03},
  year    = {2019},
  eprint  = {1904.05914},
  archivePrefix = {arXiv}
}

@article{Hatakeyama2020cls,
  author  = {Hatakeyama, K. and Matsumoto, A. and Nishimura, J. and Tsuchiya, A. and Yosprakob, A.},
  title   = {{The emergence of expanding space-time and intersecting D-branes from classical solutions in the Lorentzian type IIB matrix model}},
  journal = {PTEP},
  volume  = {2020},
  pages   = {043B10},
  year    = {2020},
  eprint  = {1911.08132},
  archivePrefix = {arXiv}
}

@article{Klinkhamer2020emu,
  author  = {Klinkhamer, F. R.},
  title   = {{On the emergence of an expanding universe from a Lorentzian matrix model}},
  journal = {PTEP},
  volume  = {2020},
  pages   = {103B03},
  year    = {2020},
  eprint  = {1912.12229},
  archivePrefix = {arXiv}
}

@article{Hirasawa2023reg,
  author  = {Hirasawa, M. and Anagnostopoulos, K. N. and Azuma, T. and Hatakeyama, K. and Nishimura, J. and Papadoudis, S. K. and Tsuchiya, A.},
  title   = {{The emergence of expanding space-time in the Lorentzian type IIB matrix model with a novel regularization}},
  journal = {PoS},
  volume  = {CORFU2022},
  pages   = {309},
  year    = {2023},
  eprint  = {2307.01681},
  archivePrefix = {arXiv}
}

@article{Ito2015pos,
  author  = {Ito, Y. and Nishimura, J. and Tsuchiya, A.},
  title   = {{Large-scale computation of the exponentially expanding universe in a simplified Lorentzian type IIB matrix model}},
  journal = {PoS},
  volume  = {LATTICE2015},
  pages   = {243},
  year    = {2016},
  eprint  = {1512.01923},
  archivePrefix = {arXiv}
}

@article{SternXu2018,
  author  = {Stern, A. and Xu, C.},
  title   = {{Signature change in matrix model solutions}},
  journal = "",
  eprint  = {1808.07963},
  archivePrefix = {arXiv},
  year    = {2018}
}

@article{Klinkhamer2020mf,
  author  = {Klinkhamer, F. R.},
  title   = {{IIB matrix model: Emergent spacetime from the master field}},
  journal = "",
  eprint  = {2007.08485},
  archivePrefix = {arXiv},
  year    = {2020}
}

@article{Klinkhamer2020pts,
  author  = {Klinkhamer, F. R.},
  title   = {{IIB matrix model: Extracting spacetime points}},
  journal = "",
  eprint  = {2008.01058},
  archivePrefix = {arXiv},
  year    = {2020}
}

@article{Klinkhamer2020rbb,
  author  = {Klinkhamer, F. R.},
  title   = {{Regularized big bang and IIB matrix model}},
  journal = "",
  eprint  = {2009.06525},
  archivePrefix = {arXiv},
  year    = {2020}
}

@article{Klinkhamer2021mfe,
  author  = {Klinkhamer, F. R.},
  title   = {{A first look at the master-field equation of the IIB matrix model}},
  journal = "",
  eprint  = {2105.05831},
  archivePrefix = {arXiv},
  year    = {2021}
}

@article{BattistaSteinacker2023,
  author  = {Battista, E. and Steinacker, H. C.},
  title   = {{One-loop effective action of the IKKT model for cosmological backgrounds}},
  journal = "",
  eprint  = {2310.11126},
  archivePrefix = {arXiv},
  year    = {2023}
}

@article{SteinackerTran2023lor,
  author  = {Steinacker, H. C. and Tran, T.},
  title   = {{Spinorial description for Lorentzian $\mathfrak{hs}$-IKKT}},
  journal = "",
  eprint  = {2312.16110},
  archivePrefix = {arXiv},
  year    = {2023}
}

@article{Anagnostopoulos2026cl,
  author  = {Anagnostopoulos, K. N. and Azuma, T. and Hirasawa, M. and Nishimura, J. and Papadoudis, S. and Tsuchiya, A.},
  title   = {{The emergence of (3+1)-dimensional expanding spacetime from complex Langevin simulations of the Lorentzian type IIB matrix model with deformations}},
  journal = "",
  eprint  = {2604.19836},
  archivePrefix = {arXiv},
  year    = {2026}
}

@article{Anagnostopoulos2026susy,
  author  = {Anagnostopoulos, K. N. and Azuma, T. and Hirasawa, M. and Nishimura, J. and Tsuchiya, A. and Yamamori, N.},
  title   = {{Impact of supersymmetry on the dynamical emergence of the spacetime in the type IIB matrix model with the Lorentz symmetry ``gauge fixed''}},
  journal = "",
  eprint  = {2604.25564},
  archivePrefix = {arXiv},
  year    = {2026}
}

@article{Steinacker:2011wb,
  author        = {Steinacker, Harold},
  title         = {Split noncommutativity and compactified brane solutions in matrix models},
  journal       = {Prog. Theor. Phys.},
  volume        = {126},
  pages         = {613--636},
  year          = {2011},
  doi           = {10.1143/PTP.126.613},
  eprint        = {1106.6153},
  archivePrefix = {arXiv},
  primaryClass  = {hep-th}
}

@article{Iso:2001mg,
  author        = {Iso, Satoshi and Kimura, Yusuke and Tanaka, Kanji and Wakatsuki, Kazunori},
  title         = {Noncommutative gauge theory on fuzzy sphere from matrix model},
  journal       = {Nucl. Phys. B},
  volume        = {604},
  pages         = {121--147},
  year          = {2001},
  doi           = {10.1016/S0550-3213(01)00173-0},
  eprint        = {hep-th/0101102},
  archivePrefix = {arXiv}
}

@article{Steinacker:2026inv,
  author        = {Steinacker, Harold C.},
  title         = {Quantum spacetime and gravity from the {IKKT} matrix model: an invitation},
  journal       = "",
  year          = {2026},
  eprint        = {2609.08955},
  archivePrefix = {arXiv},
  primaryClass  = {hep-th},
  note          = {arXiv:2609.08955}
}

@unpublished{Vogel1999,
  author = {Vogel, Pierre},
  title  = {The universal {L}ie algebra},
  year   = {1999},
  note   = {Preprint, Universit\'e Paris VII; available at \url{https://webusers.imj-prg.fr/~pierre.vogel/grenoble-99b.pdf}}
}

\end{document}